\documentclass[12pt]{article}\usepackage[]{graphicx}\usepackage[]{xcolor}
\makeatletter
\def\maxwidth{ %
  \ifdim\Gin@nat@width>\linewidth
    \linewidth
  \else
    \Gin@nat@width
  \fi
}
\makeatother

\definecolor{fgcolor}{rgb}{0.345, 0.345, 0.345}

\usepackage{framed}
\makeatletter
 {\par\unskip\endMakeFramed%
 \at@end@of@kframe}
\makeatother

\definecolor{shadecolor}{rgb}{.97, .97, .97}
\definecolor{messagecolor}{rgb}{0, 0, 0}
\definecolor{warningcolor}{rgb}{1, 0, 1}
\definecolor{errorcolor}{rgb}{1, 0, 0}
\newenvironment{knitrout}{}{} 

\usepackage{alltt}

\usepackage{natbib}
\newcommand{\citeasnoun}{\citet}
\renewcommand{\cite}{\citep}
\usepackage{amsthm}

\newcommand{\code}[1]{{\tt #1}}

\newtheorem{theorem}{Theorem}
\newtheorem{lemma}[theorem]{Lemma}
\theoremstyle{definition}
\newtheorem{definition}[theorem]{Definition}
\newtheorem{observation}[theorem]{Observation}
\newtheorem{procedure}[theorem]{Procedure}

\usepackage{amsmath}
\usepackage{amssymb}
\renewcommand{\S}{\mathcal{S}}
\newcommand{\C}{\mathcal{C}}
\newcommand{\K}{\mathcal{K}}

\title{Implementing a Metropolis sampler on decomposable graphs using 
a variety of ways to represent the graph}

\author{Alun Thomas
\footnote{295 Chipeta Way, Salt Lake City UT 84108, USA.
alun.thomas$@$utah.edu.}
\\
Division of Epidemiology\\
Department of Internal Medicine\\
University of Utah\\
}
\IfFileExists{upquote.sty}{\usepackage{upquote}}{}
\begin{document}

\maketitle

\section*{Abstract}

We describe the implementation of the Giudici-Green Metropolis
sampling method for decomposable graphs using a variety of 
structures to represent the graph. These comprise
the graph itself, the Junction tree, the Almond tree and
the Ibarra clique-separator graph. For each structure, we describe
the process for ascertaining whether adding or deleting a specific
edge results in a new graph that is also decomposable, and the 
updates that need to be made to the structure if the edge
perturbation is made.
For the Almond tree and Ibarra graph these procedures are novel.
We find that using the graph itself is 
generally at least competitive in terms of computational
efficiency for a variety of graph distributions,
but note that the other structures may allow and suggest
samplers using different perturbations with lower rejection rates
and/or better mixing properties. The sampler has applications
in estimating graphical models for systems of multivariate Gaussian or
Multinomial variables.

\section*{Keywords}
Junction tree, Almond Tree, Ibarra clique-separator graph,
structural model estimation.

\newpage
\section{Introduction}

\subsection{The Giudici-Green sampler}

\citeasnoun{Giudici+Green:99} introduced a simply stated Metropolis 
\cite{Metropolis+al:53} method for sampling from the space of decomposable graphs
under a specified probability distribution $\pi()$. The method is as follows:
\begin{itemize}
\item
Select at random two distinct elements $x$ and $y$ from the vertex set 
$V$ of the 
decomposable graph $G=G(V,E)$ with edge set $E \subseteq V \times V$.
\item
If $(x,y)$ is an edge of $G$ and disconnecting $x$ and $y$ results in a 
new graph $G^-$ which is decomposable, accept a move to $G^-$,
that is disconnect $x$ and $y$, with probability 
\begin{equation}
\max \left \{ 1, \frac{\pi(G^-)}{\pi(G)} \right \}.
\label{accgminus}
\end{equation}
\item
If $(x,y)$ is not an edge of $G$ and connecting $x$ and $y$ results in a
new graph $G^+$ which is decomposable, accept a move to $G^+$,
that is connect $x$ and $y$, with probability
\begin{equation}
\max \left \{ 1, \frac{\pi(G^+)}{\pi(G)} \right \}.
\label{accgplus}
\end{equation}
\end{itemize}

(Note that in this introduction, in order to facilitate exposition,
we assume that the reader is familiar with
decomposable graphs, cliques, separators, junction trees and so on and,
therefore, anticipate the definitions given below.)

The effectiveness of the algorithm relies on two important features.
The first is that the probability distribution $\pi()$ from with we
sample is required to be of the form
\begin{equation}
\pi(G) = \frac { \prod_{C \in \C} \phi(C)} {\prod_{S \in \S} \phi(S) } 
\label{structmark}
\end{equation}
where $\C$ is the set of cliques of $G$, $\S$ is the collection of
its separators, and $\phi(A)$ is positive valued {\em potential} 
function defined on all vertex sets $A \subseteq V$.
Such distributions were considered by \citeasnoun{Byrne+Dawid:15} and termed
{\em structurally Markov} distributions. 
We follow this requirement here,
although extension to {\em weak structurally Markov} distributions
\cite{Green+Thomas:18} is straightforward.
Under this assumption, computation of the acceptance probabilities
(\ref{accgminus}) and (\ref{accgplus}) 
requires computing $\phi()$ for only 4 vertex sets determined by $x$,
$y$, and the structure of $G$ local to $x$ and $y$. 
Thus, computing the acceptance probabilities
is done in constant time regardless of the size of $V$.

When estimating the graph structure of a graphical model for a 
set of random variables using a 
Bayesian framework, $\pi()$ is 
the posterior distribution on the graph space given a set of multivariate 
observations marginalized over the possible parameter values. 
\citeasnoun{Dawid+Lauritzen:93} show that for a broad class of 
conjugate prior distributions for sets of Gaussian or Multinomial
parameters the marginal graph posterior is structurally Markov whenever the
graph prior is.

The second feature is that we need to determine, as efficiently as possible,
whether or not $G^-$ or $G^+$ is decomposable. 
\citeasnoun{Giudici+Green:99} achieve this by maintaining, in addition 
to the current state of $G$, the current state of a {\em junction tree}
representation of $G$. A junction tree is a tree in which the vertices represent
the cliques of $G$ and the edges represent the separators. While a 
junction tree representation can be readily obtained from $G$, it is
not uniquely determined, however, the algorithm does not depend on the
specific junction tree representation used, and in fact switches between
equivalent representations.
Given this structure we can determine what we will call the 
{\em legality} of a perturbation
to a decomposable graph $G$, that is whether the perturbation results 
in a new graph that is also decomposable, as follows:
\begin{itemize}
\item
If $(x,y)$ is an edge of $G$, disconnecting $(x,y)$ is legal iff
there is a unique clique $C_{xy}$ of $G$ that contains both $x$ and $y$
\cite{Frydenberg+Lauritzen:89}.
\item
If $(x,y)$ is not an edge of $G$, connecting $(x,y)$ is legal iff
there is a clique $C_x \ni x$ and a clique $C_y \ni y$
that are adjacent is some junction tree representation of $G$
\cite{Giudici+Green:99}.
\end{itemize}

\subsection{Other characterizations of decomposable graphs}

The junction tree has proved to be a very effective characterization
of a decomposable graphs both for graph sampling, and hence inference
on graph structure when sampling from a posterior, 
and for propagating information through known graphs using
the {\em forward-backward} methods of graphical modeling
\cite{Lauritzen+Spiegelhalter:88}.
However, other characterizations have been proposed and
considered. These include the {\em Almond tree}
\cite{Almond+Kong:91}, the {\em Ibarra graph} or {\em clique-separator graph}
\cite{Ibarra:09} both of which we review below, and of course there is 
the graph itself usually
represented by a specification of the vertex adjacency sets.
While the the Giudici-Green sampler using the junction tree
has been difficult to beat computationally, there are other reasons to consider
these alternative representations.

\begin{itemize}
\item
The Giudici-Green algorithm proposes and rejects illegal perturbations 
which raises the question of whether we can limit proposals to only 
legal moves and whether a different representation makes this easier.
\item
Some structural queries may be easier with other representations. For instance,
determining whether a set of vertices is a separator of $G$ is equivalent
to asking whether it is an edge of a junction tree. However, it is also
equivalent to asking whether the set is a vertex of either the 
Almond tree or the Ibarra graph which may be easier to determine.
\item
Having a different structure to represent the graph raises the possibility
of perturbing the structure rather than the graph itself in a sampling
scheme. That is, can we sample
the space of representations of decomposable graphs rather than space
of graphs? For instance, \citeasnoun{Green+Thomas:13} described sampling
decomposable graphs using a Markov chain on junction trees. Other 
representations may suggest other perturbations that could improve
mixing properties.
\item
Because junction tree representations are not unique,
the \citeasnoun{Green+Thomas:13} junction tree sampler
requires the enumeration of the junction tree representations of a graph
which is a substantial computation.
Similarly, Almond trees are not uniquely determined, however, Ibarra
graphs are. Thus, a sampler on Ibarra graphs could potentially be simpler
and faster
as it would not be necessary to enumerate and 
correct for multiple representations.
\item
Sampler mixing properties might also be improved by proposing
multiple edge changes in a single perturbation. For example,
\citeasnoun{Green+Thomas:13}, also described a multi edge junction 
tree sampler.
Could other structures make determining the legality of multiple edge
perturbations more efficient?
\item
In large graphs it may be possible to make multiple local perturbations
independently and simultaneously and hence speed up a sampler using 
multiple processors in parallel. Are there structures that allow the 
identification of perturbations that can be validly made in parallel?
\end{itemize}

While these issues are active research questions, to date only the
junction tree representation has been used for sampling decomposable
graphs.
The work presented here describes how to determine the legality of 
a single edge perturbation using a each of these representations and how
the data structures need to be updated when a perturbation is accepted.
These we view as necessary first steps to fully investigating and exploiting
these alternative approaches.
We begin by describing the process using only the graph itself and follow
by reviewing the sampler as described by
\citeasnoun{Giudici+Green:99}.
We then consider the Almond tree the 
definition of which we modify slightly to make the edges directed.
Finally we deal with the Ibarra graph which was originally 
defined to have some directed
and some undirected edges, but which we again modify slightly to give
direction to all edges.

We believe that the formulations that allow the Metropolis sampling scheme to be
done with reference to the graph itself, the Almond tree and the Ibarra
graph are novel.

\clearpage
\section{Preliminaries}

\subsection{Decomposable graphs}

For a more complete treatment of graphical modeling and decomposable graphs
see \citeasnoun{Lauritzen:96}.

Let $G=G(V,E)$ be an undirected graph with vertex set $V$ and
edge set $E\subseteq V \times V$. A subset $U \subseteq V$ defines
an {\em induced subgraph} of $G$ with vertex set $U$ and edge
set $E \cap U \times U$. The induced subgraph of $U$ is
{\em complete} if
its edge set is $U \times U$.
$U$ is a {\em clique} if it is maximally complete, that is,
it induces a complete subgraph in $G$ and there
is no subset $W$ such that $U \subset W \subseteq V$ that also induces
a complete subgraph of $G$.
The set of all cliques of $G$ is denoted by $\C$.

\begin{definition}[Running intersection property]
A graph $G$ has the {\em running intersection property} if its cliques can
be ordered
$(C_1, C_2, \ldots C_c)$ such that
\begin{equation}
\mbox{if} \ \ S_i = C_i \cap \bigcup_{j=1}^{i-1} C_j
\ \ \mbox{then} \ \
S_i \subset C_k
\ \ \mbox{for some} \ \
k < i.
\label{running}
\end{equation}
The collection of sets $\S = \{ S_i \}$ are called the {\em separators} of $G$.
\end{definition}
A clique ordering that has this property is called a {\em perfect ordering}
and while $G$ will typically have many perfect orderings,
both the set of cliques and the collection of separators are uniquely
defined. Because of maximality, the cliques must be distinct, however,
in general the separators are not. The number of times a set appears as
a separator is called its {\em multiplicity}.

\begin{definition}[Decomposable graph]
A graph is {\em decomposable} if and only if it has the running intersection
property.
\end{definition}
In what follows the graphs are decomposable unless stated otherwise.
Figure \ref{rawpic} gives an example of a decomposable graph
from \citeasnoun{Thomas+Green:09}.

\begin{figure}[htb]
\caption{An example of a decomposable graph. \label{rawpic}}
\bigskip
\centerline{\includegraphics[width=5.5in]{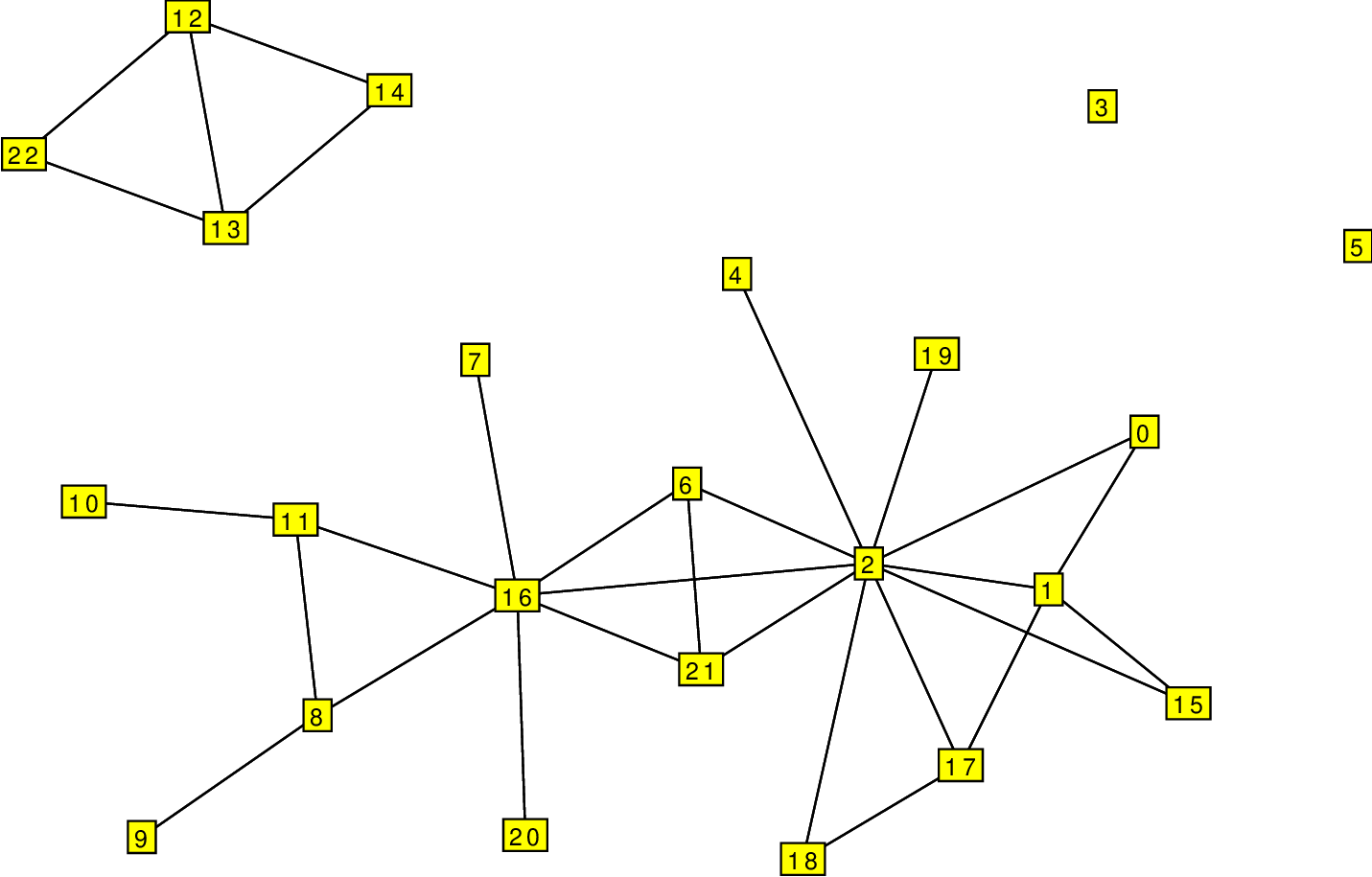}}
\end{figure}

Related to the perfect clique order is the {\em vertex elimination}
order, that is, an ordering of the vertices with the property that
the intersection of the adjacency set of a vertex with the vertices 
that follow it
in order is complete. A vertex elimination order also exists iff the graph is 
decomposable. 
Given a vertex elimination order it is straightforward to specify
an {\em edge elimination} order, that is, an order of the edges such that
at each stage, when an edge is removed, the resulting graph 
is decomposable.
We can, therefore, construct any decomposable graph by
adding its edges to the empty graph in reverse elimination order using
only legal moves.
Since we will describe below how to update each representation of 
a decomposable graph when a legal edge is added, we do not also give
other explicit methods for their construction, however, note that
\citeasnoun{Tarjan+Yannakakis:84} and \citeasnoun{Ibarra:09} give 
such methods for junction trees and Ibarra graphs respectively.
Instead, we need only specify in addition how to represent 
the empty graph, which is always decomposable, using each characterization.

\begin{definition}
\mbox{ }\\
\begin{itemize}
\item 
A {\em path} between $x$ and $y$ is a sequence of vertices
$(x = u_1, u_2, \ldots u_n = y)$ such that each $(u_i,u_{i+1})$ is 
an edge in $G$.
The {\em length} of the path is $n-1$, the number of edges in it.
\item
A {\em cycle}  is a path with $x = y$.
\item
A {\em chord} of the cycle 
$(x = u_1, u_2, \ldots u_n = x)$
is an edge $(u_i,u_j)$ that is not in the cycle itself.
A cycle is {\em unchorded} if it has no chords.
\item
A vertex set $A$ is said to {\em separate} vertices $x$ and $y$ in $G$ if 
every path from $x$ to $y$ in $G$ contains at least one vertex in $A$.
\end{itemize}
\end{definition}

A further defining property of decomposable graphs is that they contain
no unchorded cycles of length greater than 3, and decomposable graphs
are therefore sometimes called {\em chordal} or {\em triangulated} graphs.
We will exploit this below to prove results by contradiction, that is, we will
show that, under some circumstances, cycles or 4 or more vertices 
exist and that, hence, a graph is not decomposable.

\subsection{Graph searches}

In what follows we will make much use of searches in a graph to find paths
between vertices, or to establish that no such path
exits. The generic form of the searches we use are as follows:
\begin{procedure}[Generic path search]
\mbox{ }\\
\begin{enumerate}
\item
Identify an initial vertex $a$.
\item
Create an empty queue of vertices and $a$ to it.
Mark $a$ as {\em explored}.
\item
Create a vertex to vertex map to enable backtracking.
\item
Until the queue is empty:
\begin{itemize}
\item
Remove vertex $b$ from the head of the queue and process it.
\item
For each unexplored neighbour $c$ of $b$:
\begin{itemize}
\item 
Mark $c$ as explored.
\item
If $c$ satisfies a continuation criterion, add $c$ to the queue and
map $c$ to $b$ in the backtracking map.
\end{itemize}
\end{itemize}
\end{enumerate}
\end{procedure}

This format can accommodate standard searches such as breadth first, by
using a first-in-first-out queue, or depth first, by using a last-in-first-out
queue.
Junction trees, Almond trees, and Ibarra graphs 
are all graphs in which vertices comprise sets of vertices of the
original decomposable graph. Many of our searches on these structures 
will use a queue of sets that prioritizes
the entries in increasing or decreasing order of size: 
a {\em lightest neighbour first} or {\em heaviest neighbour first} search
respectively.
We will also mention below {\em maximum spanning tree algorithms},
such as Prim's algorithm \cite{Jarnik:30}, which prioritize the search
based on edge properties rather than vertex properties and do not fit
into this generic format, however, these will only be used 
conceptually and do not need to be implemented.

To mark vertices as explored we can maintain the set of explored vertices
and if this is implemented, for example, by a hash map, adding, removing,
and checking for membership can all be done in constant time.

To establish whether a path exists between two vertices, we set one of
these as the initial vertex and continue until either the other is
reached, or the search is exhausted.

If we need a specific path between two vertices, not just to know whether
such a path exists, we can use the map to track back from the 
last vertex reached to starting point. If the direction of the path is
relevant, care should be taken to choose the initial vertex appropriately.

\clearpage
\section{Giudici-Green using the decomposable graph itself}

For any pair of vertices $x,y \in V$ let $S_{xy}$, be the set of their common
neighbours in decomposable $G(V,E)$, that is, 
the intersection of their adjacency sets,
and
let $S_x = S_{xy} \cup \{x\}$,
$S_y = S_{xy} \cup \{y\}$.
$C_{xy} = S_{xy} \cup \{x,y\}$.

\subsection{Disconnections}

A simple condition for the legality of disconnections in terms 
of the common neighbours of the end vertices is well
know \cite{Eriksen:96,Culbertson+al:21} and also lets us obtain
the Metropolis acceptance probability.
A simple proof from first principles is given here.

\begin{theorem}[Legal disconnections]
Any pair of connected vertices $x,y$ can be legally disconnected
iff $S_{xy}$ is complete.
\label{disconnections}
\end{theorem}

\begin{proof}
Let $G$ be a decomposable graph in which $(x,y)$ is an edge and let
$G^-$ be the graph obtained from $G$ by disconnecting $(x,y)$.

If $S_{xy}$ is not complete in $G$ 
$\exists u,v \in S_{xy}$ such that
$u$ and $v$ are not connected in $G$ and, hence, not connected in $G^-$. 
Since they are both common neighbours of $x$ and $y$ in $G$ and $G^-$,
$(x,u,y,v,x)$ is an unchorded 4-cycle in $G^-$ which is, therefore,
not decomposable.

If $S_{xy}$ is complete in $G$, assume for the purpose of
contradiction that $G^-$ is not decomposable. $G^-$ must
then contain an unchorded k-cycle for $k>3$ that is chorded by the 
addition of the edge $(x,y)$.
If $k > 4$, it requires more than one edge to chord the
cycle which contradicts the decomposability of $G$.
So, $G^-$ contains the 4-cycle $(x,u,y,v,x)$ for some vertices $u$ and $v$.
But $u$ and $v$ are then unconnected common neighbours of $x$ and $y$
in $G$ which contradicts the completeness of $S_{xy}$.
\end{proof}

When this condition is met, 
$C_{xy}$ is the unique clique of 
$G$ that contains both $x$ and $y$ that is specified
in the criterion of \cite{Frydenberg+Lauritzen:89}
and it is straightforward to show that these criteria are equivalent.

\begin{definition}[Clique enabling a disconnection]
When disconnecting $x$ and $y$ is a legal move, $C_{xy}$ is a uniquely
defined clique and we say that $C_{xy}$ {\em enables} the disconnection.
\end{definition}
Note that the set of legal disconnections can be partitioned over
their enabling cliques.

\begin{observation}[Acceptance probability for a disconnection] 
If disconnecting $(x,y)$ is a legal move, the appropriate
Metropolis acceptance probability is 
\begin{equation}
\max 
\left \{ 
1, \frac{\phi(S_{x}) \phi(S_y)}{\phi(S_{xy})\phi(C_{xy})} 
\right \}.
\end{equation}
\end{observation}
\begin{proof}
Simple restatement of the result of \citeasnoun{Giudici+Green:99}.
\end{proof}

Constructing $S_{xy}$ and checking for completeness 
is straightforward
using the adjacency sets of the vertices in $G$, but is quadratic in
$|S_{xy}|$ and may be computationally expensive for large dense graphs.

\subsection{Connections}
For legality of connections and the Metropolis acceptance 
probability we have the following novel observation.

\begin{theorem}[Legal connections]
Any pair of unconnected vertices $x,y \in V$ can be legally connected
iff 
$S_{xy}$ separates $x$ and $y$ in $G$.
\end{theorem}

\begin{proof}
Let $G$ be a decomposable graph in which $(x,y)$ is not an edge and let
$G^+$ be the graph obtained from $G$ by connecting $(x,y)$.

If $S_{xy}$ does not separate $x$ and $y$ in $G$, there must be paths
in $G$ that do not pass through any common neighbour.
The shortest of these paths must be of length greater than 2, otherwise
the sole intermediate vertex would be a common neighbour, and there can be
no other edges connecting the intermediate vertices, otherwise a 
shorter path would exist. Connecting $(x,y)$ in $G^+$ then creates an
unchorded cycle of length greater than 3, and $G^+$ is not decomposable.

If $S_{xy}$ separates $x$ and $y$ in $G$, assume for the purpose of 
contradiction that $G^+$ is not decomposable. $G^+$ must then contain
an unchorded k-cycle of length $k >= 4$ that is broken by the 
disconnection of $(x,y)$.
$G$ must therefore contain a path $(x,u,\ldots,v,y)$ of length $k>=3$.
Since $S_{xy}$ separates $x$ and $y$ in $G$, this path must contain
a vertex $w$ which is a common neighbour of $x$ and $y$. We then have
a contradiction since  $(x,w)$ and/or $(y,w)$ will be a chord of 
the k-cycle in $G^+$.
\end{proof}

$S_{xy}$ is always complete when $x$ and $y$ are connected,
because otherwise there must be an unchorded 4-cycle, and when
the above condition is met, it is a separator of $G$.

\begin{definition}[Separator enabling a connection]
When connecting $x$ and $y$ is a legal move, $S_{xy}$ is a uniquely
defined separator and we
say that $S_{xy}$ {\em enables} the connection. 
\end{definition}
Note that the set of legal connections can be partitioned over their enabling
separators.

\begin{observation}[Acceptance probability for a connection] 
If connecting $(x,y)$ is a legal move, the
Metropolis acceptance probability is 
\begin{equation}
\max 
\left \{ 
1, \frac {\phi(S_{xy})\phi(C_{xy})} {\phi(S_{x}) \phi(S_y)}
\right \}.
\end{equation}
\end{observation}
\begin{proof}
Simple restatement of the result of \citeasnoun{Giudici+Green:99}.
\end{proof}

Checking for legality of separation requires a graph search
to establish that no path exists.
In the worst case this can involve
traversing the
whole graph. To some extent, however, this might be mitigated by the 
following observation.

\begin{lemma}
If $S_{xy}$ is complete and there exits a path from $x$ to $y$ that
doesn't include any vertex in $S_{xy}$, then every vertex in  that
path must be connected to every vertex in $S_{xy}$.
\label{connecttosxy}
\end{lemma}
\begin{proof}
Suppose that $u$ is a vertex in the path from $x$ to $y$ and that
$v \in S_{xy}$ with $u$ not connected to $v$. 
Then, $x,v,y,\ldots,u,\ldots,x$ is a cycle that must contain a cycle of
length at least 4 that is unchorded, the missing chord being $(u,v)$,
which contradicts the decomposability of $G$.
\end{proof}

This allows us to restrict the path searching to an
induced subgraph of $G$.

\begin{theorem}[Restricted path search]
Let $N_{xy}$ be the set of vertices not in $S_{xy}$ that are connected
to every vertex in $S_{xy}$. Let $G_{N_{xy}}$ be the subgraph of 
$G$ induced by $N_{xy}$. 
$S_{xy}$ separates $x$ and $y$ in $G$ iff $x$ and $y$ if there
is no path between $x$ and $y$ in $G_{N_{xy}}$.

In other words, connecting $x$ and $y$ in $G$ is legal iff 
$x$ and $y$  are in 
different components of $G_{N_{xy}}$.
\end{theorem}
\begin{proof}
This follows directly from lemma \ref{connecttosxy} as $G_{N_{xy}}$ 
contains precisely those vertices that paths between $x$ and $y$ must
pass through.
\end{proof}

In many cases $G_{N_{xy}}$ can be far smaller than $G$ significantly
reducing the workload of establishing separation, however, 
identifying and restricting the search to $G_{N_{xy}}$ may 
involved significant computation.
In the
worst case, notably when $S_{xy}$ is empty and a separator of $G$, 
there is no saving.

\subsection{The sampler using the graph itself}

We can now restate the Giudici-Green sampler directly in terms
of the graph itself. 
In doing so, since checking the random Metropolis acceptance 
criterion is typically
far faster than checking for legality of the perturbation, we 
reverse the order of the criteria and proceed to the second
only when the first is met.

We note that, in the obvious way and without recourse to other
representations, we can
track the probability, up to 
a constant factor, of the current graph.
In an analogous fashion we can also track its clique set
and separator collection.
Let $G^{\emptyset} = G(V,E = \emptyset)$ be the trivial or empty graph.

\begin{procedure}[Giudici-Green with the graph only]
\label{thegraphitself}
\mbox{}\\
\begin{enumerate}
\item
\begin{itemize}
\item
Initialize $G = G^{\emptyset}$, 
\item
Set $\pi(G) = \pi(G^{\emptyset}) = \frac{\prod_{v \in V} \phi( \{ v \})}{\phi(\emptyset)^{n-1}}$.
\item
Set 
$\C = \C(G^{\emptyset} = \{ \{ v_1 \}, \{v_2\}, \ldots \{v_n\} \}$
\item
Set
$\S = \S(G^{\emptyset} = \{ \emptyset, \emptyset, \ldots \emptyset \}$,
that is a list of length $n-1$ in which the every element is the empty set.
\end{itemize}
\item
Randomly select two distinct vertices $x,y \in V$
and find $S_{xy}, S_x, S_y,$ and $C_{xy}$.
\label{loop}
\item
Generate a random variate $U \sim$ Uniform$(0,1)$.
\item
For $x$ and $y$ connected:
if 
\begin{displaymath}
	U <= \frac{\phi(S_x)\phi(S_y)}{\phi_(S_{xy}) \phi(C_{xy})}
\ \ \mbox{ and $C_{xy} \in \C$,}
\end{displaymath}
\begin{itemize}
\item
Disconnect $(x,y)$.
\item
Set $\pi(G) = \pi(G) \frac {\phi(S_x) \phi(S_y)} {\phi(C_{xy})\phi(S_{xy})}$.
\item 
Remove $C_{xy}$ from $\C$.
\item 
Add an instance of $S_{xy}$ to $\S$.
\item 
If $S_x \in \S$ remove one instance of it from $\S$,
otherwise add it to $\C$.
\item 
If $S_y \in \S$ remove one instance of it from $\S$,
otherwise add it to $\C$.
\end{itemize}
\item
For $x$ and $y$ unconnected:
if
\begin{displaymath}
	U <= \frac {\phi_(S_{xy}) \phi(C_{xy})} {\phi(S_x)\phi(S_y)}
\ \ \mbox{ and $S_{xy} \in \S$ and $S_{xy}$ separates $x$ and $y$,}
\end{displaymath}
\begin{itemize}
\item
Connect $(x,y)$.
\item
Set $\pi(G) = \pi(G) \frac {\phi(C_{xy})\phi(S_{xy})}  {\phi(S_x) \phi(S_y)}$.
\item 
Add $C_{xy}$ to $\C$.
\item 
Remove an instance of $S_{xy}$ from $\S$.
\item 
If $S_x \in \C$ remove it from $\C$,
otherwise add an instance of it it to $\S$.
\item 
If $S_y \in \C$ remove it from $\C$,
otherwise add an instance of it it to $\S$.
\end{itemize}
\item
Repeat from step \ref{loop}.
\end{enumerate}
\end{procedure}

Note that in step 4, testing for $C_{xy} \in \C$ is equivalent to 
testing that $S_{xy}$ is complete, and this takes $O(|C_{xy}|)$ time
whereas testing for the completeness of $S_{xy}$ using $G$ directly
take time of $O(|S_{xy}|^2)$.
Also, in step 5, the pre check that $S_{xy} \in \S$ can avoid a graph
search when $S_{xy}$ is not a separator of $G$ and hence cannot 
separate $x$ and $y$.
Similarly, tracking $\pi(G)$ is not necessary to run the sampler, but it may be
a useful statistic to monitor.

We can implement this procedure with two functions: one that
takes $x$, $y$ and $C_{xy}$, checks that $C_{xy}$ enables the disconnection
and, when it does, makes the disconnection by updating the data structure;
and one that takes $x$, $y$ and $S_{xy}$, checks that $S_{xy}$ enables
the connection, and when it does, makes the connection.
These, together with a specification of the data structure to match an initial
state, are sufficient to run the sampler.

When implementing the the sampler using junction trees, etc,
we can use the  alternative graph
representations as replacements for  the graph itself, 
or use them in conjunction with it.
\citeasnoun{Giudici+Green:99}, in effect, used the junction tree as a 
replacement. There is, however, good reason to keep the
raw graph representation,
for instance, if the graph is being displayed during the sampling process
or if summary statistics about its structure are required as output these
may be more readily obtained from the graph itself. More importantly,
the graph itself allows quick ascertainment of $S_{xy}$ and $C_{xy}$
while each of the other representations requires a path search 
to obtain them.
In what follows we will provide methods to find $S_{xy}$ using the
alternative graph representations, but note when used as the the
only representation these  method may be combined with the updating
procedures, and bypassed completely when used as an auxiliary.

\clearpage
\section{Giudici-Green using junction trees, the original formulation}

\subsection{Definitions and properties}

We now review junction trees, their relevant properties, and the Giudici-Green
sampler.

All the alternative representations of a decomposable graph $G(V,E)$ that we 
consider below are graphs in which the vertices represent subsets
of the initial vertex set $V$. We now define a property of such graphs. 

\begin{definition}[Set graphs]
For any set $\K$ of subsets of $V$,
$\K = \{ K_1, K_2, \ldots K_m : K_i \subseteq V \}$
any graph $H$  with vertex set $\K$,
that is, 
$H = H(\K, E \subseteq \K \times \K)$, is a {\em set graph}.
\end{definition}

\begin{definition}[Junction property]
Let $A \subseteq V$ and define $H_A$ to be the subgraph of set graph $H$
induced by the vertices that contain $A$.
If $H_A$ is a single connected component for all $\{ A \subseteq V \}$, 
then $H$ is said
to have the {\em junction property}.
\end{definition}

\begin{definition}[Junction trees]
Let $G(V,E)$ be a decomposable graph with clique set $\C$.
Let $J = J(\C, E \subset \C \times \C)$ be a tree with vertex set $\C$,
so that $|E| = |\C|-1$. 
If $J$ has the junction property $J$ is said to be a {\em junction tree}.
\end{definition}

\begin{figure}[htb]
\caption{A junction tree for the graph in figure \ref{rawpic}.
\label{junctionpic}}
\bigskip
\centerline{\includegraphics[width=5.5in]{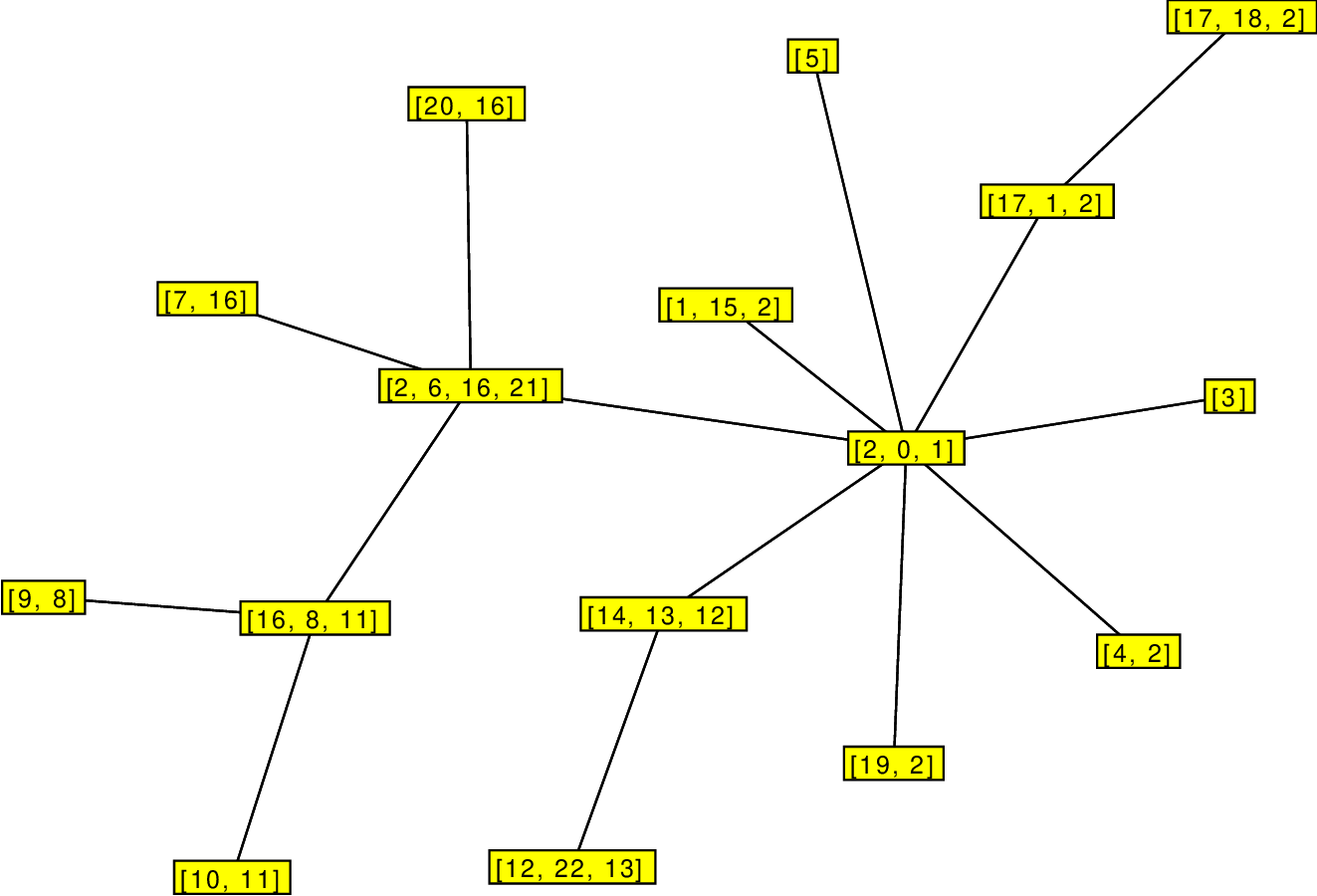}}
\end{figure}

Figure \ref{junctionpic} shows a junction tree for the decomposable
graph given in figure \ref{rawpic}.

A decomposable graph will typically have many equivalent junction tree
representations and \citeasnoun{Thomas+Green:09} give a method to 
enumerate these. Conversely, however, a junction tree uniquely determines its
decomposable graph.

Let $H^C = H^C(\C,\C \times \C)$ be a complete graph and let the edges of 
the graph be associated with the intersection of the cliques it connects.
The {\em weight} of an edge $(C_1,C_2)$ is defined as $|C_1 \cap C_2|$.
Then, $J$ is a junction tree of $G$ iff it is a heaviest spanning
tree of $H^C$, and this is sometimes used as an alternative definition.
While it is not
generally efficient to find a junction tree using
spanning tree algorithms it provides a useful alternative formulation,
in particular, it is possible at times to replace set equality and
set inclusion queries with simpler comparisons of set size. 

The method given by \cite{Tarjan+Yannakakis:84}
to construct a junction tree 
does so by finding a perfect ordering of the cliques
and then, at each step, connecting clique $C_i$ to some
$C_k, k<i$ such that
$C_k \supset  C_i \cap C_k = S_i$. 
Thus, each edge of the junction tree is associated with the intersection
of the cliques that it connects, which is a separator of $G$.


It will be useful in the following procedures to quickly identify a
clique that contains any given vertex which we will do by maintaining
a suitable vertex to clique map.
\begin{definition}[Vertex(-clique) map]
A {\em vertex-clique map}, or just {\em vertex map} is any map
that associates each vertex of $G$ to some clique of $G$ that contains it.
\end{definition}


\subsection{The sampler using junction trees}

We initialize the sampler to be the trivial graph $G^{\emptyset}$
which has cliques 
$\C_0 = \{ \{ v_i \}  \  i = 1 \ldots |V| \}$
and find a junction tree representing this state, which can be any
tree with this vertex set.

\begin{procedure}[Representing the trivial graph]
\mbox{ }\\
\begin{enumerate}
\item 
Set $ J = J(\C_0, E = \{  ( \{ v_1 \}, \{ v_{i} \}  ), 2 \leq i \leq |V| \})$.
\item
Set the vertex map by mapping $v_i$ to $\{ v_i \} \forall i$.
\end{enumerate}
\end{procedure}

Then we provide an optional method to identify $S_{xy}$ and 
for any $x, y \in V$.
\begin{procedure}[Finding $S_{xy}$]
\mbox{ }\\
\begin{enumerate}
\item
From the vertex map, find $C_x \ni x$.
\item
Search from $C_x$ until a clique $C_y \ni y$ is found.
\item
If $x \in C_y$, then $x$ and $y$ are connected in $G$.
Searching from
$C_x$ find $J_{xy}$ the connected subtree of $J$ induced by cliques
that contain $\{ x, y \}$. 
Then,
\begin{equation}
S_{xy} = \bigcup_{C\in V(J_{xy})} C - \{x,y\}
\end{equation}
where the union is taken over 
the cliques that are vertices of  $J_{xy}$.
\item
Otherwise, $x$ and $y$ are not connected in $G$.
Backtrack from $C_y$ along the path to $C_x$ and reset $C_x$ to be
the first clique in the path that contains $x$.
Then,
\begin{equation}
S_{xy} = C_x \cap C_y.
\end{equation}
\end{enumerate}
\end{procedure}

As always, 
$S_x = S_{xy} \cup \{ x \}$,
$S_y = S_{xy} \cup \{ y \}$,
$C_{xy} = S_{xy} \cup \{ x,y \}$.

\begin{procedure}[Disconnect $(x,y)$ if $C_{xy}$ enables]
\mbox{ }\\
\begin{enumerate}
\item
If $C_{xy}$ is not a vertex of $J$, $C_{xy}$ is not a clique and does not 
enable the disconnection so reject the perturbation.
Otherwise the disconnection is enabled.
\item
If $C_{xy}$ has a neighbour $C_x \supseteq S_x$ 
\begin{itemize}
\item 
disconnect $C_x$ from $C_{xy}$, 
\item
otherwise set $C_x = S_x$ and add it as a vertex of $J$.
\end{itemize}
\item
If $C_{xy}$ has a neighbour $C_y \supseteq S_y$ 
\begin{itemize}
\item 
disconnect $C_y$ from $C_{xy}$, 
\item
otherwise set $C_y = S_y$ and add it as a vertex of $J$.
\end{itemize}
\item
For each neighbour of $C_{xy}$, if it contains $x$ connect it to
$C_x$, otherwise connect it to $C_y$.
\item
Delete $C_{xy}$ from $J$.
\item
Connect $C_x$ to $C_y$.
\item
Update the vertex map by mapping each vertex in $C_x$ to $C_x$ and
each vertex in $C_y$ to $C_y$.
\end{enumerate}
\end{procedure}

\begin{procedure}[Connect $(x,y)$ if $S_{xy}$ enables]
\mbox{ }
\begin{enumerate}
\item
From the vertex map find, $C_x \ni x$.
Search from $C_x$ until a clique $C_y \ni y$ is found. 
Backtrack
from $C_y$ along the path to $C_x$ and reset $C_x$ to be the 
first clique in the path that contains $x$.
This finds the shortest path between a vertex containing $x$ and one 
containing $y$.
\item
Find in the path an edge representing a separator $S$ such that
$|S| = |S_{xy}|$. If no such edge exists, $S_{xy}$ does not
enable the connection, so reject the perturbation.
Otherwise, by the junction property, $S = S_{xy}$ and the
connection is enabled.
\item
Disconnect the vertices connected by the edge representing $S$. 
\item
Add $C_{xy}$ to the vertex set of $J$.
\item
If $|C_x| < |C_{xy}|$, 
\begin{itemize}
\item
then $C_{xy} \supset C_x$, so connect
the neighbours of $C_x$ to $C_{xy}$, and delete $C_x$ from the graph,
\item
otherwise, connect $C_x$ to $C_{xy}$.
\end{itemize}
\item
If $|C_y| < |C_{xy}|$, 
\begin{itemize}
\item
then $C_{xy} \supset C_y$, so connect
the neighbours of $C_y$ to $C_{xy}$, and delete $C_y$ from the graph,
\item
otherwise, connect $C_y$ to $C_{xy}$.
\end{itemize}
\item
Map every vertex in $C_{xy}$ to $C_{xy}$.
\end{enumerate}
\end{procedure}

\clearpage
\section{Giudici-Green using Almond trees}

\citeasnoun{Jensen:88} introduced a version of the junction tree that
included separators as intermediate vertices
along the edges connecting cliques, as such, each
separator was represented by multiple vertices corresponding to 
its multiplicity in the running intersection. 
\citeasnoun{Almond+Kong:91} modified this
formulation into a tree using a single vertex for each separator. We
slightly modify this further by assigning direction to each edge in the tree.

\begin{definition}(Almond tree)
For a decomposable graph $G$ with cliques $\C$ and separators $\S$ consider a
graph $H$ to have
\begin{itemize}
\item
vertex set $\C \cup \S$
\item
with edges $\{ (S,T) \}$ for each pair where $ S \subset T$.
\end{itemize}
Any spanning tree $A$ of $H$ with the junction property is an {\em 
undirected Almond tree}
for $G$ and is a {\em (directed) Almond tree} when its edges are
oriented from $S$ to $T$ where $S \subset T$..
\end{definition}

\citeasnoun{Almond+Kong:91} called $H$ the {\em junction graph}.
It is convenient to refer to the in-neighbours of a set to be its
{\em parents} and to the out-neighbours as its {\em children}. 
Vertices corresponding to cliques have no children in an Almond tree,
and the number of children for a separator vertex is one more than
its multiplicity.

An alternative formulation is given by \citeasnoun{Jensen+Jensen:94},
as follows.

\begin{theorem}
Assign to each edge $(S,T)$ of $H$ a weight $|S\cap T| = |S|$ if $S\subset T$.
Any heaviest spanning tree $A$ of $H$ is an undirected Almond tree for $G$.
\end{theorem}
\begin{proof}
Theorem 3 of \citeasnoun{Jensen+Jensen:94}.
\end{proof}

Note that when edge of equal weight can be chosen in constructing the
heaviest spanning tree, \citeasnoun{Jensen+Jensen:94} choose $(S,T)$
to minimize $|T|$, however, this is not necessary to satisfy the
original definition of \citeasnoun{Almond+Kong:91}.

\begin{figure}[htb]
\caption{An Almond tree for the graph in figure \ref{rawpic}.
\label{almondpic}}
\bigskip
\centerline{\includegraphics[width=5.5in]{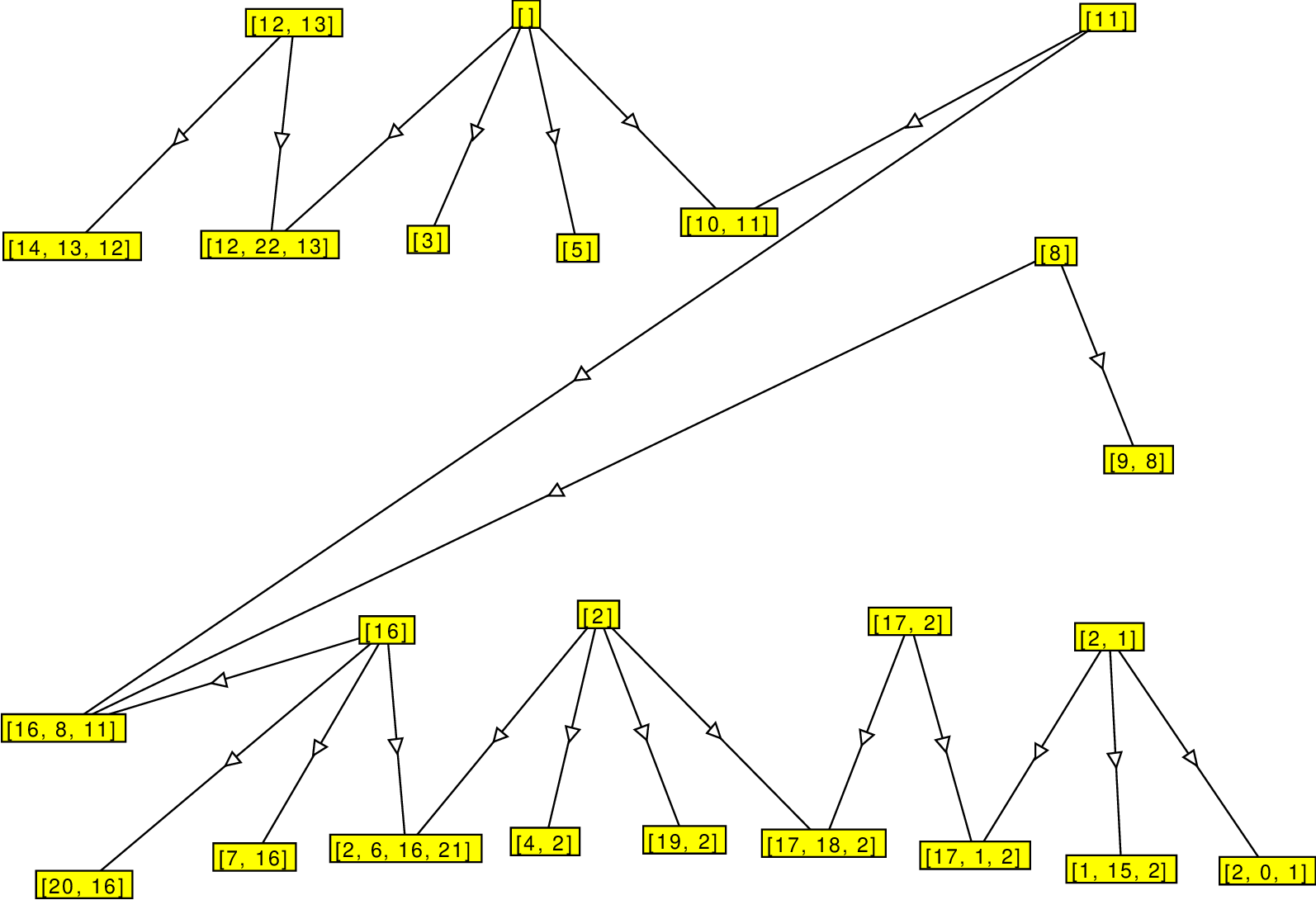}}
\end{figure}

Figure \ref{almondpic} shows an Almond tree for the decomposable
graph given in figure \ref{rawpic}.

In what follows, we will consider intermediate graphs that are not
true Almond trees because a separator has become a subset of a single
clique and needs to be removed. The following procedure achieves the.
removal of these {\em redundant} separators.

\begin{procedure}[Check for and remove a redundant separator $S$]
\mbox{}\\
\begin{enumerate}
\item
If $S$ has no children, it is a clique, not a separator, so stop here.
\item
If $S$ has more than 1 child, it is not redundant, so stop here.
\item
If $S$ has exactly 1 child, it is redundant, so remove it as follows.
\begin{itemize}
\item
Connect each parent of $S$ to its child.
\item
Remove $S$ from the graph.
\end{itemize}
\end{enumerate}
\end{procedure}

\subsection{The sampler using Almond trees}

We now define the methods to implement the Giudici-Green sampler 
using Almond trees.

\begin{procedure}[Representing the trivial graph]
\mbox{ }\\
\begin{enumerate}
\item
Set the vertex set to be $ \emptyset \cup \{  \{v_1\} \ldots \{v_n\} \}$.
\item
Connect the directed edges $\{ (\emptyset,v_i) \forall i \}$.
\item
Set the vertex map by mapping $v_i$ to $ \{ v_i \} \forall i$.
\end{enumerate}
\end{procedure}

\begin{procedure}(Disconnecting $(x,y)$ if $C_{xy}$ enables)
\mbox{ }\\
\begin{enumerate}
\item 
If $C_{xy}$ is not a vertex of the Almond tree it is not complete, the
disconnection is not legal, so stop here.
\item
If $S_{xy}$ is vertex of the tree, find the path from $C_{xy}$ to
$S_{xy}$ and disconnect $S_{xy}$ from its neighbour in the path.

\item
If $S_x$ is not a vertex in the tree, set $C_x = S_x$.
\item
Otherwise, $S_x$ is a vertex of the tree and it must be connected to $C_{xy}$.
\begin{itemize}
\item
Disconnect $(S_x,C_{xy})$.
\item
If $S_x$ is not now redundant, set $C_x = S_x$.
\item
Otherwise, set $C_x$ to be the child of $S_x$ and remove redundant $S_x$.
\end{itemize}

\item
If $S_y$ is not a vertex in the tree, set $C_y = S_y$.
\item
Otherwise, $S_y$ is a vertex of the tree and it must be connected to $C_{xy}$.
\begin{itemize}
\item
Disconnect $(S_y,C_{xy})$.
\item
If $S_y$ is not now redundant, set $C_y = S_y$.
\item
Otherwise, set $C_y$ to be the child of $S_y$ and remove redundant $S_y$.
\end{itemize}

\item
For each neighbour $S$ of $C_{xy}$ in the graph, if $S \ni x$ connect
$(S,C_x)$, otherwise connect $(S,C_y)$.
\item
Remove $C_{xy}$ from the graph.
\item
Connect $(S_{xy},C_x)$ and $(S_{xy},C_y)$.
\item
\begin{itemize}
\item 
Map ever vertex in $C_x$ to $C_x$.
\item 
Map ever vertex in $C_y$ to $C_y$.
\end{itemize}
\end{enumerate}
\end{procedure}

\begin{procedure}(Connecting $(x,y)$ if $S_{xy}$ enables)
\mbox{ }\\
\begin{enumerate}
\item
If $S_{xy}$ is not a vertex of the graph, it is not a separator and 
the connection is not legal. Stop here.
\item
From the vertex map find $C_x \ni x$, and find the path from
$C_x$ to $S_{xy}$. 
Reset $C_x$ to be last vertex in the path containing $x$ and
record $P_x$ the penultimate vertex in the path before
$S_{xy}$.
\item
From the vertex map find $C_y \ni y$, and find the path from
$C_y$ to $S_{xy}$. 
Reset $C_y$ to be last vertex in the path containing $y$ and
record $P_y$ the penultimate vertex in the path before
$S_{xy}$.
\item
If $P_x = P_y$, $S_{xy}$ does not separate $x$ and $y$ so stop here.
\item
Disconnect $(S_{xy},P_x)$ and $(S_{xy},P_y)$.
\item
If $S_x$ is a vertex of the graph, connect $(S_x,C_{xy})$, and 
remove $S_x$ if it is redundant.
\item
Otherwise, connect $(S_x,C_{xy})$ and $(S_x,C_x)$.
\item
If $S_y$ is a vertex of the graph, connect $(S_y,C_{xy})$, and 
remove $S_y$ if it is redundant.
\item
Otherwise, connect $(S_y,C_{xy})$ and $(S_y,C_x)$.
\item
Connect $(S_{xy},C_{xy})$. If $S_{xy}$ is redundant, remove it.
\item
Map every vertex in $C_{xy}$ to $C_{xy}$.
\end{enumerate}
\end{procedure}

The procedure to find $S_{xy}$ is the same as for Junction trees.

\clearpage
\section{Giudici-Green using Ibarra graphs}

\subsection{Definitions and properties}

The {\em clique-separator graph} for a decomposable graph $G$ was defined by 
\citeasnoun{Ibarra:09} to have vertices corresponding to both 
the cliques and separators of $G$, with directed edges between separators,
and undirected edges connecting separators and cliques. We slightly
modify the definition to make all edges directed, and refer
to the resulting structure as an {\em Ibarra graph}.

\begin{definition}[Ibarra's clique-separator graph]
For a decomposable graph $G$ with cliques $\C$ and separators $\S$
define the Ibarra graph, $I$, to have 
\begin{itemize}
\item
vertex set $\C \cup \S$
\item
with directed edges $\{ (S,T) \}$ from each pair
where  $S \subset T$, and there is no separator $U$ such that 
$S \subset U \subset T$. 
\end{itemize}
\end{definition}

Note that the Ibarra graph has a single vertex for each separator regardless
of its multiplicity. We also note that since, clearly, $I$ is a
connected spanning subgraph of the junction graph, any heaviest edge
spanning tree of $I$ is an Almond tree, but there are Almond trees that use
edges not present in the Ibarra graph.
Almond trees that satisfy
the additional requirement of \citeasnoun{Jensen+Jensen:94}
will be heaviest spanning trees of the Ibarra graph, but not
every heaviest spanning tree of an Ibarra graph will satisfy 
this requirement.

\begin{figure}[htb]
\caption{The Ibarra graph for the graph in figure \ref{rawpic}.
\label{ibarrapic}}
\bigskip
\centerline{\includegraphics[width=5.5in]{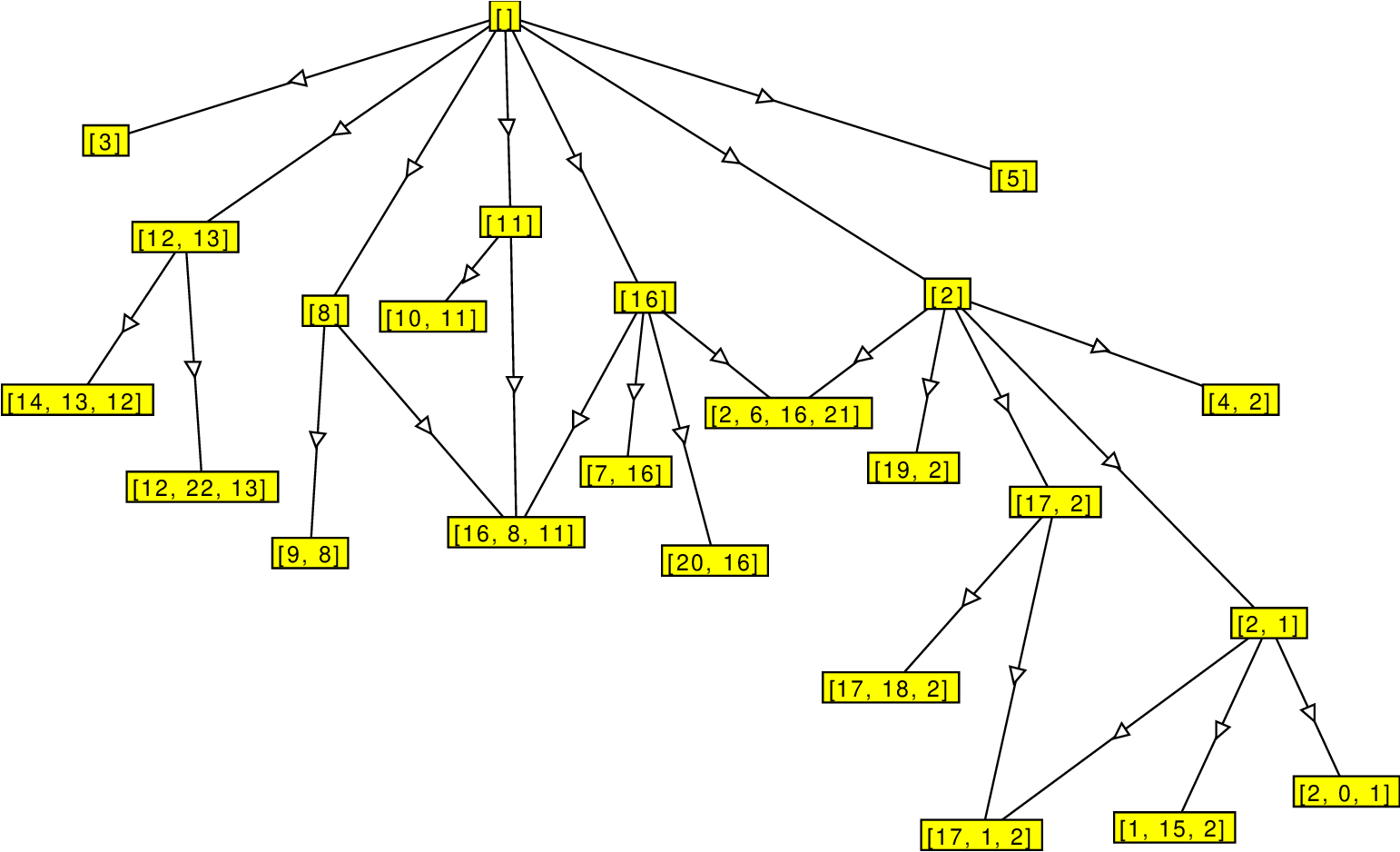}}
\end{figure}

Clearly, since the edges of $I$ are uniquely defined by the cliques and
separators, there is a unique Ibarra graph for any decomposable graph, and 
conversely, the Ibarra graph contains a complete specification of the
cliques and so uniquely defines the corresponding decomposable graph.

Figure \ref{ibarrapic} shows the Ibarra graph for the decomposable
graph given in figure \ref{rawpic}.

\begin{theorem}(Junction property)
The Ibarra graph has the junction property.
\end{theorem}
\begin{proof}
For any set $A$ if no clique or separator contains $A$ the 
junction property holds trivially.

Now consider any separator that contains $A$.
Since this separator is an intersection of cliques, there must be at
least 2 cliques that contain it and hence contain $A$.
Moreover there are directed paths from the separator
to each of these cliques and each vertex along the paths also contains $A$.

Now consider the cliques that contain $A$. If there is only 1, then it
is the only vertex in the graph containing $A$, and the junction property
holds.
Now suppose there are $m$ cliques containing $A$. Label them
$C_1^A, \ldots C_m^A$ such that the order of the cliques in the
perfect ordering of the cliques of $G$ is preserved.
For any clique $C_i^A$ if $A \subset C_i^A \bigcup_{j=1}^{i-1}C_j$ there is a 
clique $C_k, k < i$ such $A \subset C_k$ so $C_k = C_l^A$ for some
$l < i$. Also there is a separator $S_i^A = C_i^A \cap C_l^A \supset > A$
and, hence, there is a path from $S_i^A$ to each of $C_i^A$ and $C_l^A$
where each vertex on the path also contains $A$, and hence a path
between $C_i^A$ and $C_l^A$ in which all the vertices contain $A$.
Applying this for $i=2,\ldots m$ shows that the the vertices containing
$A$ form a single connected sub graph of $I$ and junction property holds.
\end{proof}

The Ibarra graph is a special case of a {\em set inclusion graph}
\cite{Huang+al:19}.
Finding all 
subsets of a set in the graph is straightforwardly done by recursively 
following parental edges, that is, a set's subsets are its 
{\em ancestors} in the graph. Similarly, a set's supersets are its
{\em descendants}.

Vertices corresponding to cliques have no children, but will have
parents except in the case when $G$ comprises a fully connected single clique. 

Vertices corresponding
to separators always have children, but may or may not have parents.
The subgraph of $I$ induced by the descendants of a separator $S$ must 
have at least 2 components, these being the subgraphs separated by $S$.
The multiplicity of $S$ is one less than the number of components in 
this induced subgraph.

\subsection{Required sub-procedures}

In each case when we need to add a vertex $X$ to an Ibarra graph we will
be able to readily identify a superset $C$ that is already in the graph.
We can use this to simplify and speed up the process. 
We will refer to this as {\em adding $X$ above $C$}.

\begin{procedure}[Adding a set $X$ above a known superset $C$]
\mbox{ }\\
\begin{enumerate}
\item
Find all the ancestors of $C$, that is all its subsets, and 
sort them from largest to smallest.
\item
\label{ups}
For each set $A$ in the list:
\begin{itemize}
\item
If $A \subset X$, connect the edge $(A,X)$, and remove $A$ and all
its ancestors, that is all its subsets, from the list.
\item
Otherwise, just remove $A$ from the list.
\end{itemize}
\item
Starting from $C$, find all the sets in the graph that contain $X$, and 
sort them from smallest to largest. Because of the junction property, the
supersets of $X$ will be a single connected component of the graph, so this
search can be localized.
\item
For each set $A$ in the list, connect the edge $(X,A)$, and remove
$A$ and all its descendants, that is all its supersets, from the list.
\label{downs}
\item
Disconnect any edges that connect a parent of $X$ to a child of $X$.
\end{enumerate}
\end{procedure}

A more general method for adding a  set to a set inclusion graph 
can be had by, at step \ref{ups}, starting  with a list sorted
from largest to smallest of
all the vertices in the graph, and at step \ref{downs} starting with 
the reverse of this list and checking for set inclusion before making
the connection.

As with the case of the Almond tree,
in the course of updating an Ibarra graph to accommodate 
connecting or disconnecting an edge we will make intermediate graphs 
with redundant separators. Thse are  
separators with descent sets that have a single component. These
need to be removed to make a true Ibarra graph.

\begin{procedure}[Checking for and removing a redundant separator $S$]
\mbox{ }\\
\begin{enumerate}
\item
Find all the descendants of $S$.
\item
Starting with any descendant, and traversing only edges that connect 
descendants, visit all vertices reachable from the starting point.
\item
If all its descendants have been reached, then $S$ is redundant
and needs to be removed from the graph. Otherwise $S$ separates 
the components among its descendants so stop here.
\item
Find the parents of $S$, the children of $S$, and then remove $S$ from the
graph.
\item
For each parent $P$ and child $C$, if $P$ is not  an ancestor of $C$ in
the graph with $S$ deleted, connect the edge $(P,C)$.
\end{enumerate}
\end{procedure}

\subsection{The sampler using Ibarra graphs}
We can now define the methods sufficient to implement the Giudici-Green
sampler using Ibarra graphs.

\begin{procedure}[Representing the trivial graph]
\mbox{ }\\
\begin{enumerate}
\item
Set the vertex set to be $ \emptyset \cup \{  \{v_1\} \ldots \{v_n\} \}$.
\item
Connect the directed edges $\{ (\emptyset,v_i) \forall i \}$.
\item
Set the vertex map by mapping $v_i$ to $ \{ v_i \} \forall i$.
\end{enumerate}
\end{procedure}

The above procedure is identical to that for Almond trees.

\begin{procedure}[Disconnect $(x,y)$ if $C_{xy}$ enables]
\mbox{ }\\
\begin{enumerate}
\item
If $C_{xy}$ is not a vertex of the graph, it is not complete and
the disconnection is not 
legal so stop here.
\item
\begin{itemize}
\item
If $S_{xy}$ is not a vertex of the graph, add it above superset $C_{xy}$.
\item
If $S_x$ is not a vertex of the graph, add it above superset $C_{xy}$.
\item
If $S_y$ is not a vertex of the graph, add it above superset $C_{xy}$.
\end{itemize}
\item
Delete $C_{xy}$ from the graph.
\item
\begin{itemize}
\item 
If $S_x$ is redundant, remove it from the graph and reset $S_x$ to be 
one of its descendants.
\item 
If $S_y$ is redundant, remove it from the graph and reset $S_x$ to be 
one of its descendants.
\end{itemize}
\item
Map all vertices in $S_x$ to $S_x$, and all vertices of $S_y$ to $S_y$.
\end{enumerate}
\end{procedure}

\begin{procedure}[Connect $(x,y)$ if $S_{xy}$ enables]
\mbox{ }\\
\begin{enumerate}
\item
\begin{itemize}
\item
If $S_{xy}$ is not a vertex of the graph, it is not a separator, and
hence does not separate $x$ and $y$, the connection is not legal
so stop here.
\item
Otherwise, from the vertex map find $C_y \ni y$.
\item
Starting from $C_y$, following edges in either direction, always
visiting the biggest reachable set at each step, search until the 
first set $C_x \ni x$ is found.
\item
Backtrack along the path and reset $C_y$ to be the first set containing
$y$ reached. This gives the shortest path from a set containing
$x$ to a set containing $y$.
\item
If this path does not contain the vertex $S_{xy}$, $S_{xy}$ does not 
separate $x$ and $y$ so stop here.
\end{itemize}
\item
\begin{itemize}
\item
If $S_x$ is not a vertex, add it above superset $C_x$.
\item
If $S_y$ is not a vertex, add it above superset $C_y$.
\end{itemize}
\item
Add $C_{xy}$ to the vertex set and connect the edges
$(S_x,C_{xy})$ and $(S_y,C_{xy})$.
\item
\begin{itemize}
\item
If $S_x$ is redundant remove it from the graph.
\item
If $S_y$ is redundant remove it from the graph.
\item
If $S_{xy}$ is redundant remove it from the graph.
\end{itemize}
\item
Map every vertex in $C_{xy}$ to $C_{xy}$.
\end{enumerate}
\end{procedure}

The procedure to find $S_{xy}$ is essentially the same as for 
Junction trees and Almond trees except that the search from $C_y$ to $C_x$
always visits the biggest reachable set first, and the search must 
allow for multiple paths between vertices as the Ibarra graph is not
in general a tree.

\clearpage
\section{Results}

All of the above methods have been implemented by the author in a 
Java program called \code{DecosDemo} available from github at
{
\par
\centering
\code{https://github.com/alun-thomas/decos}.
\par
}
This program demonstrates sampling decomposable graphs from distributions that control 
for the density of edges and clique sizes.
These are all examples of structurally Markov
distributions that might plausibly be used 
as priors in Bayesian structural model estimation.
The user can either choose one representation to use for the sampling and
compare computational times, or view all graph representations 
simultaneously for a single sampling run. 
This program was used to obtain the 
graph drawings in this work, and the results presented where we compare 
computational times for each representation for a variety of runs over
graphs of different sizes and different probability densities.
The sampling distributions we considered were:
\begin{itemize}
\item
Uniform over decomposable graphs.
\item
Uniform over decomposable graphs with maximum clique size 3.
This was done by setting the clique potential to be
\begin{eqnarray}
	\phi(C) & = & 1 \ \ \mbox{if} \ |C| \leq 3 \nonumber \\
		& = & 0 \ \ \mbox{otherwise}.
\end{eqnarray}
\item
With probability favouring smaller graphs using an edge penalty,
specifically
\begin{equation}
\pi(G(V,E)) \propto e^{-\alpha |E|}
\end{equation}
which was done using clique potential defined by
\begin{equation}
\log (\phi(C)) = \frac{\alpha |C|(|C|-1)}{2}.
\end{equation}
We did this for $\alpha = 1$ and $\alpha = 2$.
\end{itemize}

For each distribution we show the computational times for sampling graphs
of in a range of sizes from 100 to 2000. As indicators for whether equilibrium
sampling has been reached, we show plots for the number of edges in the
graph and the acceptance probability which we 
define as proportion of graphs which pass the distribution criterion
and are also decomposable.
Also shown is the 1 millionth graph of size 100 sampled under each 
distribution.
By using the same random seed each time so that the representations
deal with exactly the same sequence of proposals and, hence, sampled graphs.
All program runs were done on a laptop with
11th Gen Intel Core i9-11950H 2.60GHz processors.

\begin{figure}[htb]
\caption{
\label{timesunif}
Computational times by number of vertices when making 
1M and 10M samples uniformly over decomposable graphs
using a variety of graph representations.
Black = the graph itself, red = junction tree,
green = Almond tree, blue = Ibarra graph.
}
\bigskip
\begin{knitrout}
\definecolor{shadecolor}{rgb}{0.969, 0.969, 0.969}\color{fgcolor}
\includegraphics[width=\maxwidth]{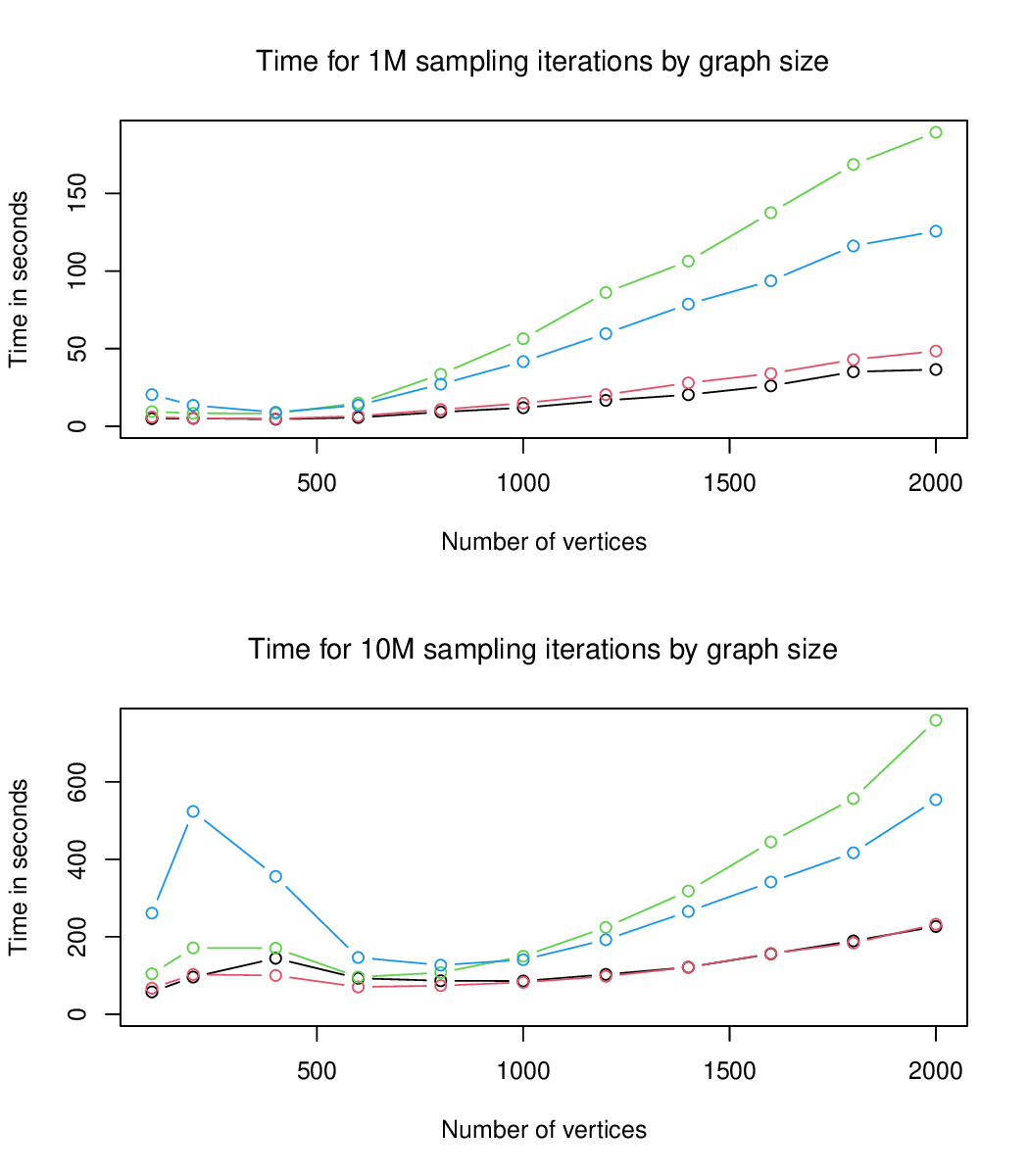} 
\end{knitrout}
\end{figure}

\begin{figure}[htb]
\caption{
\label{edgesunif}
Plots of the number of edges and acceptance probability
by iteration
when making
10M
samples
uniformly over decomposable graphs
for a variety of graph sizes.
Black = 100 vertices, red = 500, green = 1000, blue = 1500, cyan = 2000.
}
\bigskip
\begin{knitrout}
\definecolor{shadecolor}{rgb}{0.969, 0.969, 0.969}\color{fgcolor}
\includegraphics[width=\maxwidth]{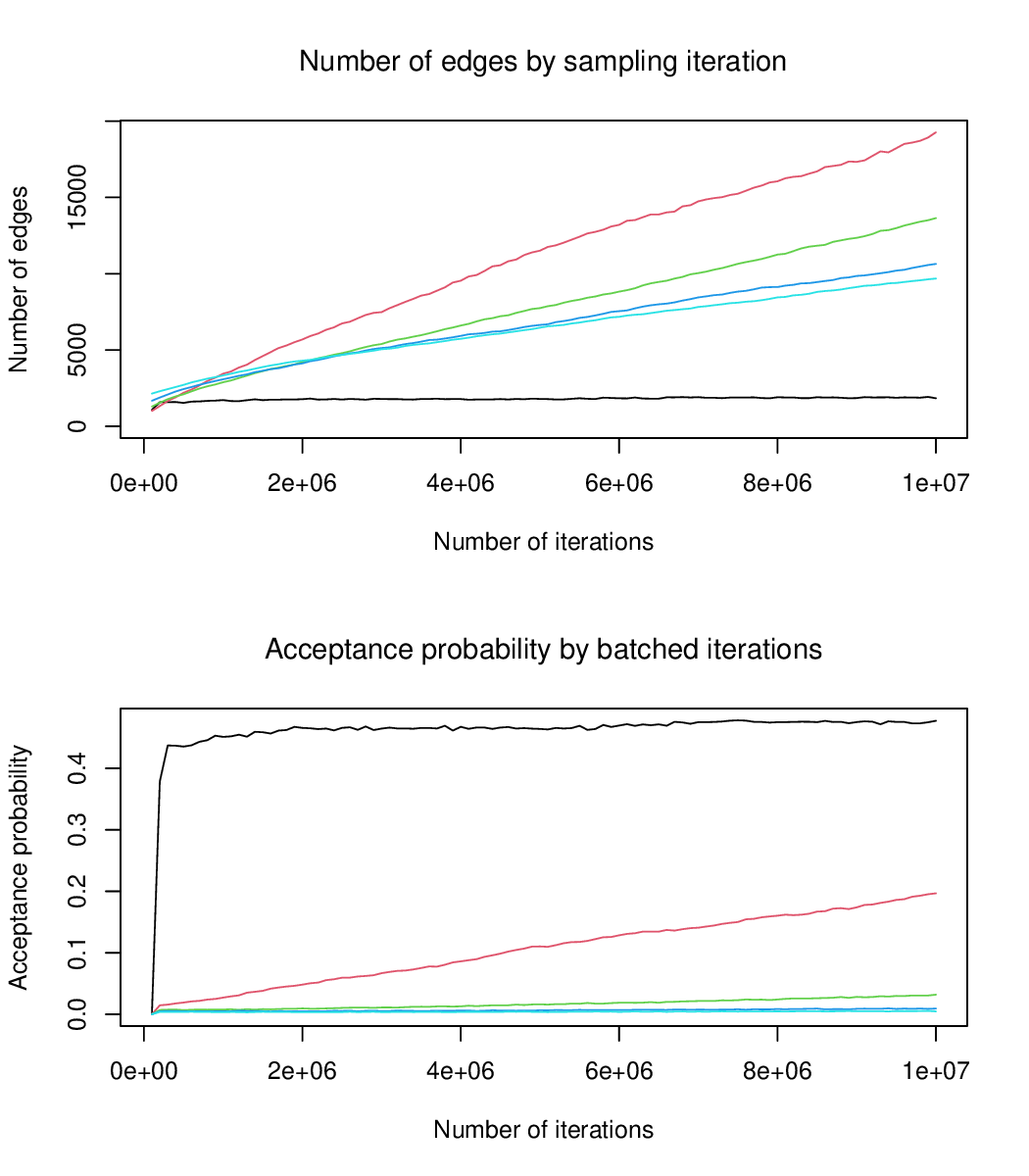} 
\end{knitrout}
\end{figure}

\begin{figure}[htb]
\caption{The 1000000th graph with 100 vertices sampling uniformly over decomposable graphs.
Top left: graph; top right: junction tree; bottom left: Almond tree;
bottom right: Ibarra graph.
\label{eguniform}}
\bigskip
\centerline{
\includegraphics[width=3.0in]{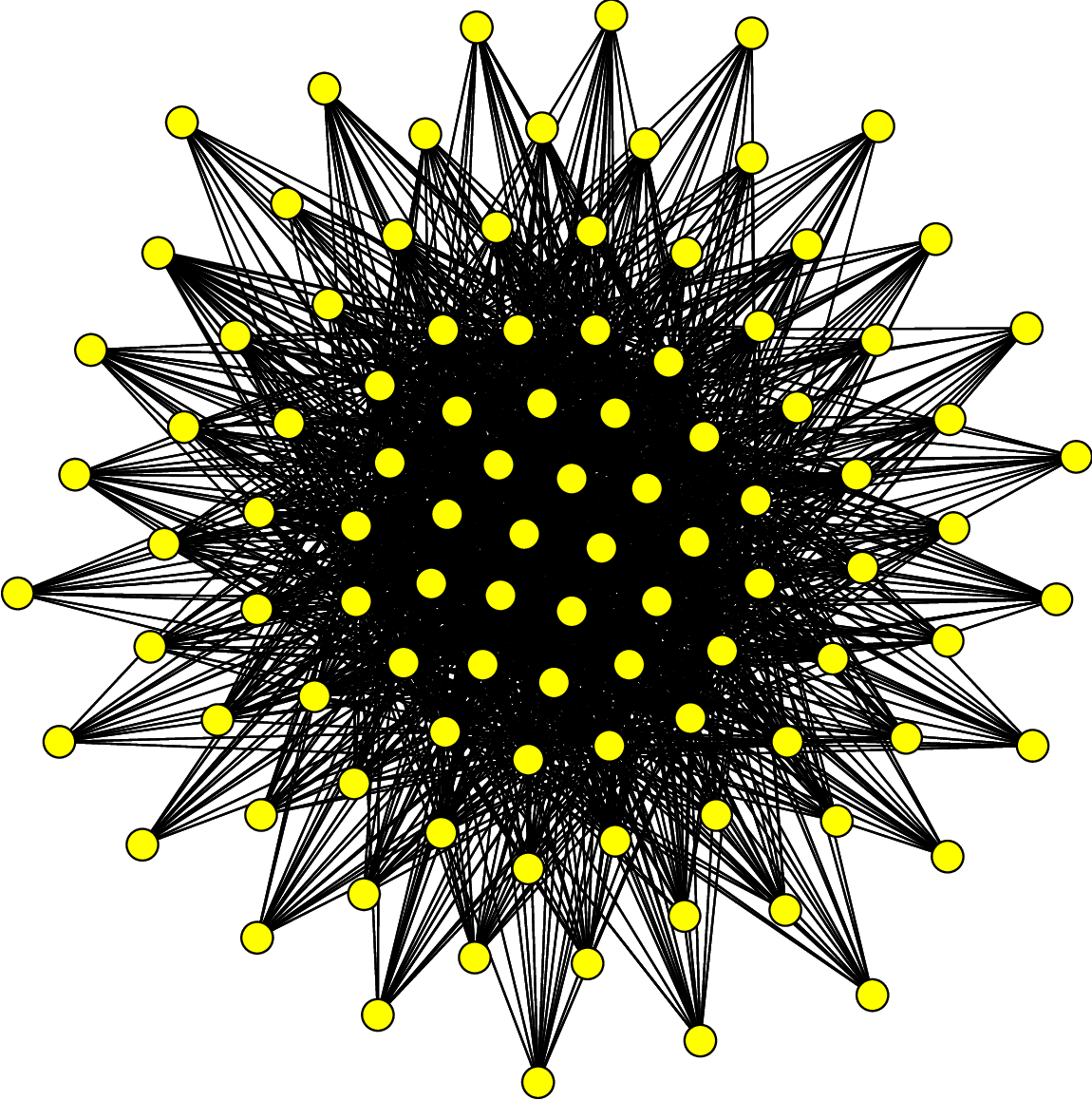}
\includegraphics[width=3.0in]{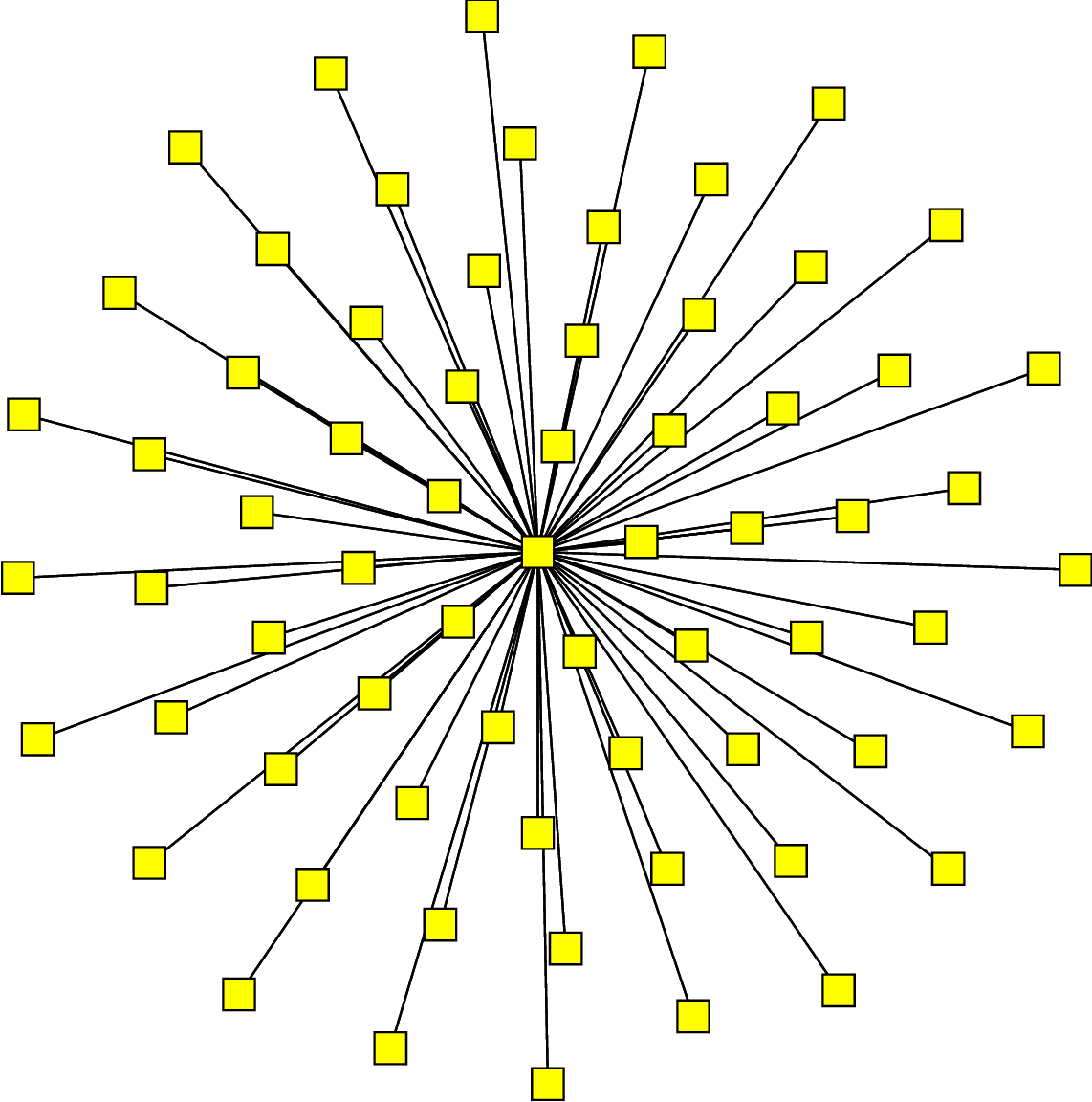}
}
\centerline{
\includegraphics[width=3.0in]{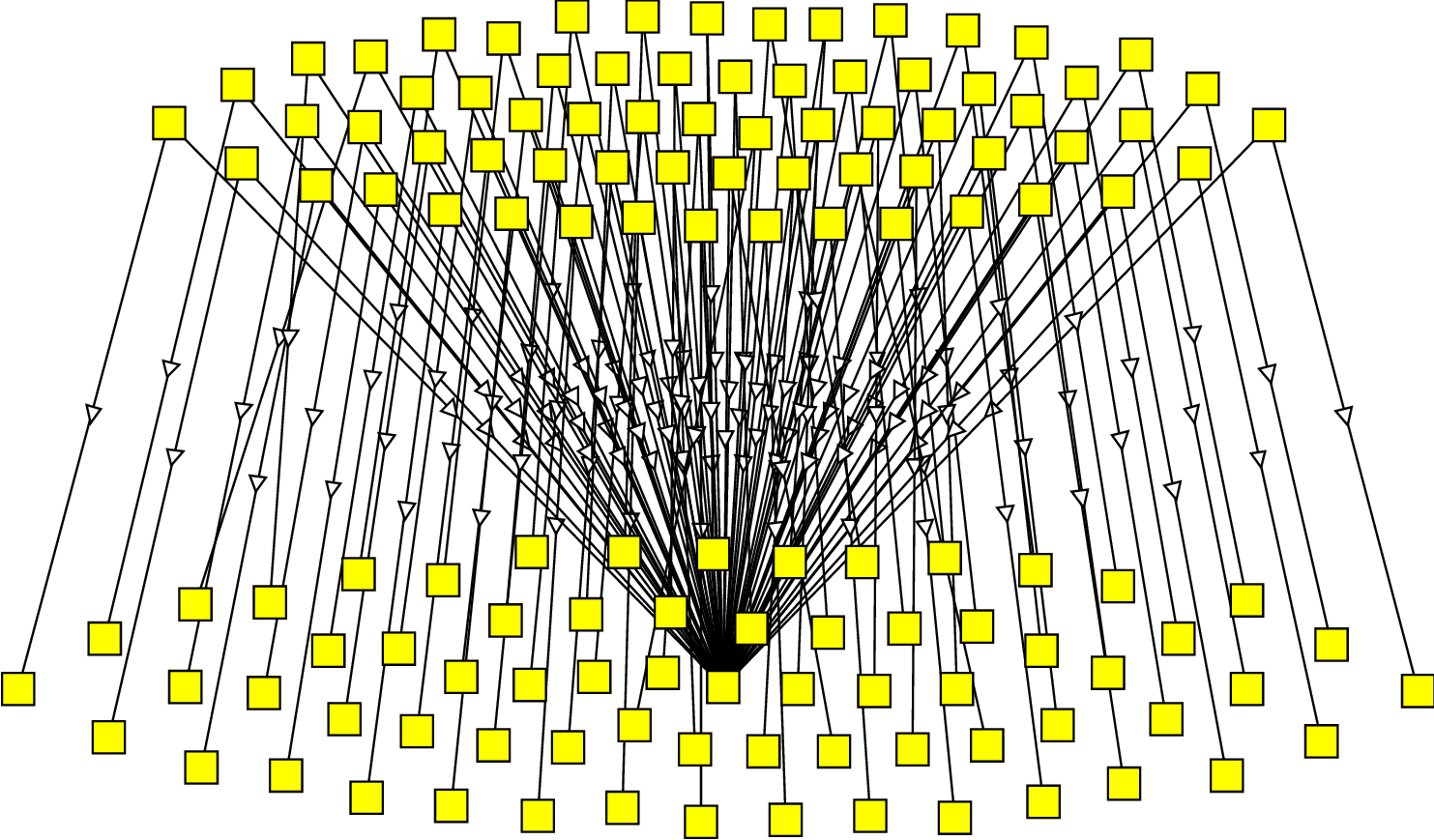}
\includegraphics[width=3.0in]{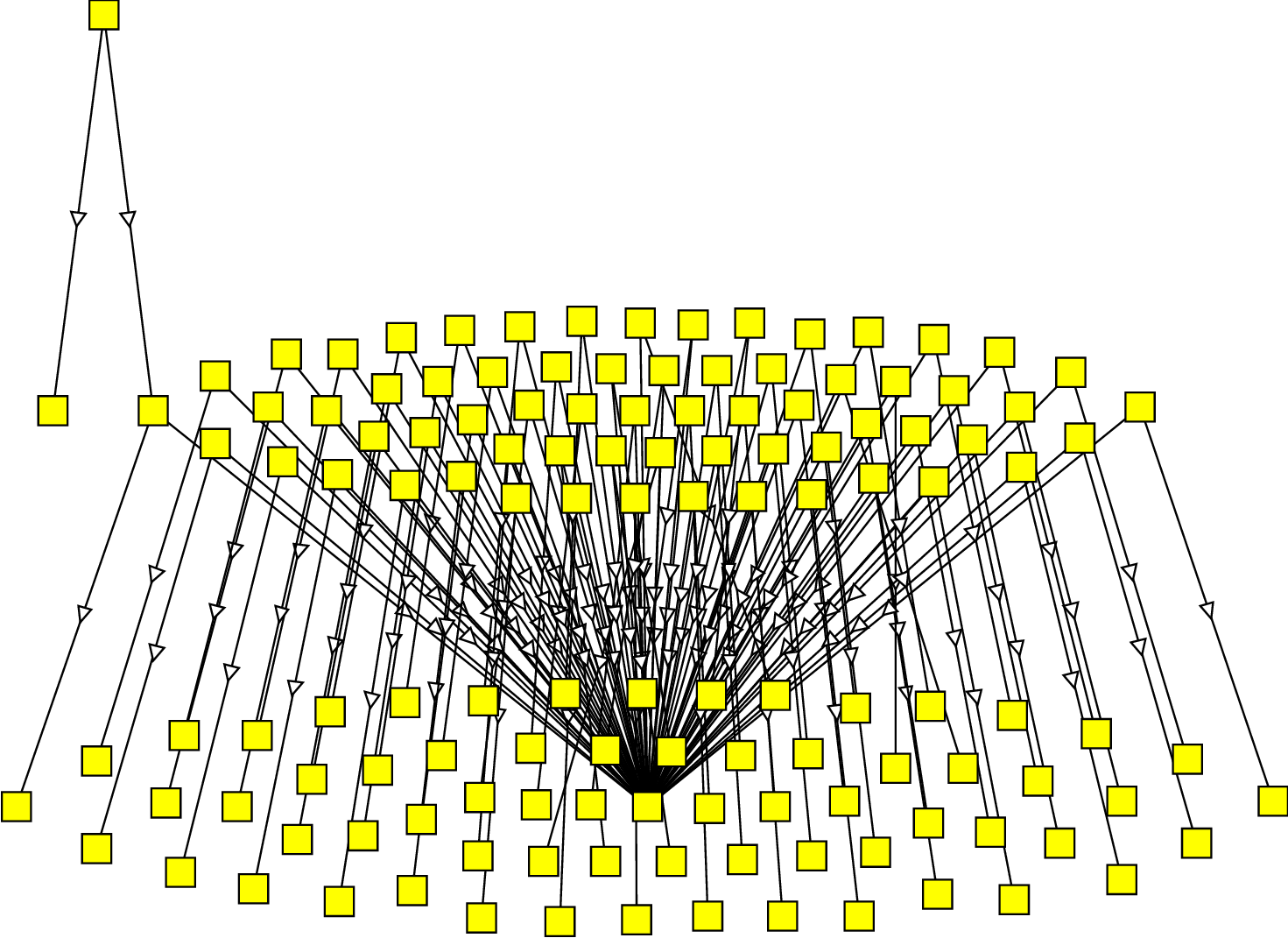}
}
\end{figure}

\begin{figure}[htb]
\caption{
\label{timesmax3}
Computational times by number of vertices when making 
1M and 10M
samples
uniformly over decomposable graphs
with maximum clique size 3
using a variety of graph representations.
Black = the graph itself, red = junction tree,
green = Almond tree, blue = Ibarra graph.
}
\bigskip
\begin{knitrout}
\definecolor{shadecolor}{rgb}{0.969, 0.969, 0.969}\color{fgcolor}
\includegraphics[width=\maxwidth]{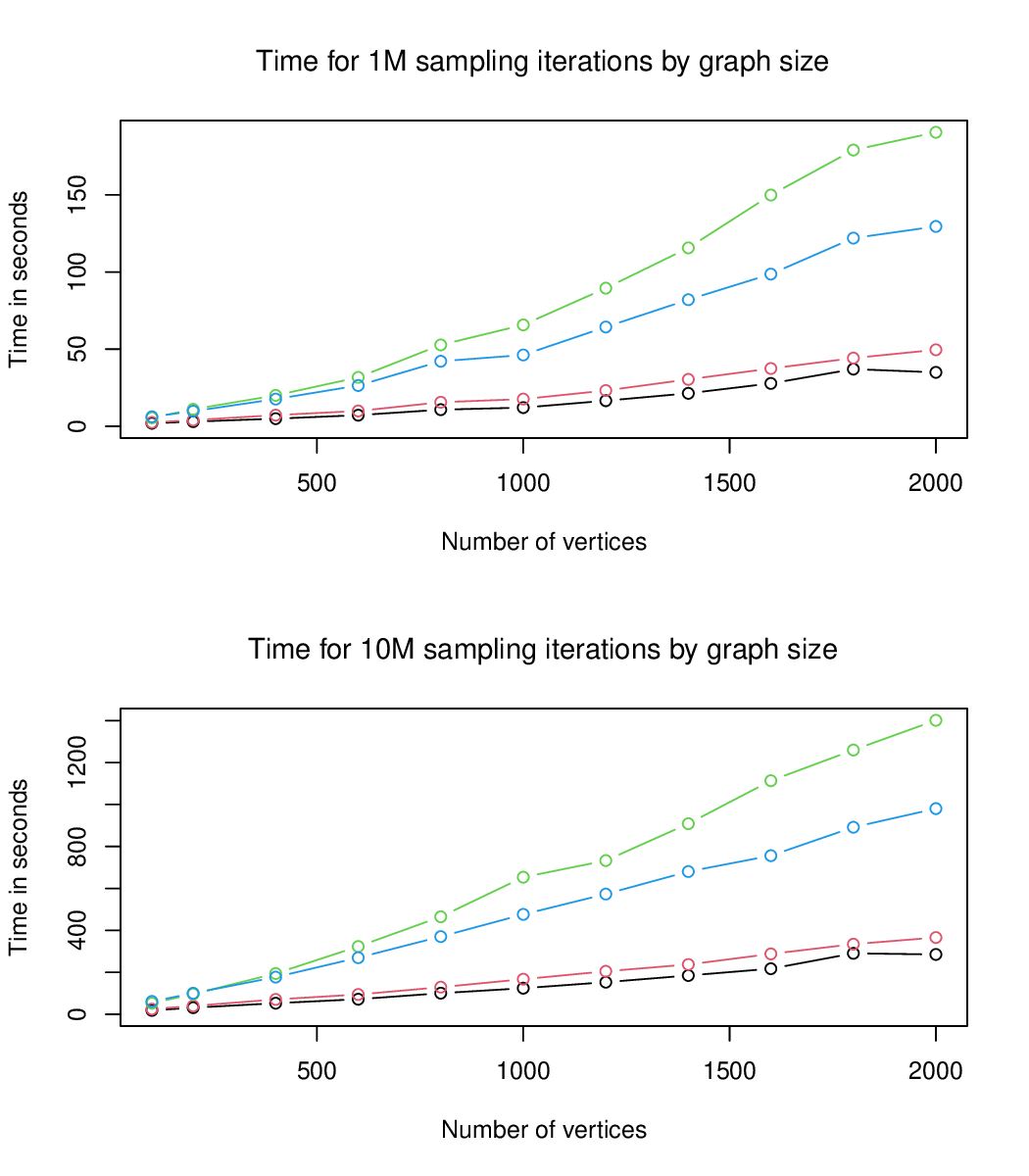} 
\end{knitrout}
\end{figure}

\begin{figure}[htb]
\caption{
\label{edgesmax3}
Plots of the number of edges and acceptance probability
by iteration
when making
10M
samples
uniformly over decomposable graphs with maximum clique size 3
for a variety of graph sizes.
Black = 100 vertices, red = 500, green = 1000, blue = 1500, cyan = 2000.
}
\bigskip
\begin{knitrout}
\definecolor{shadecolor}{rgb}{0.969, 0.969, 0.969}\color{fgcolor}
\includegraphics[width=\maxwidth]{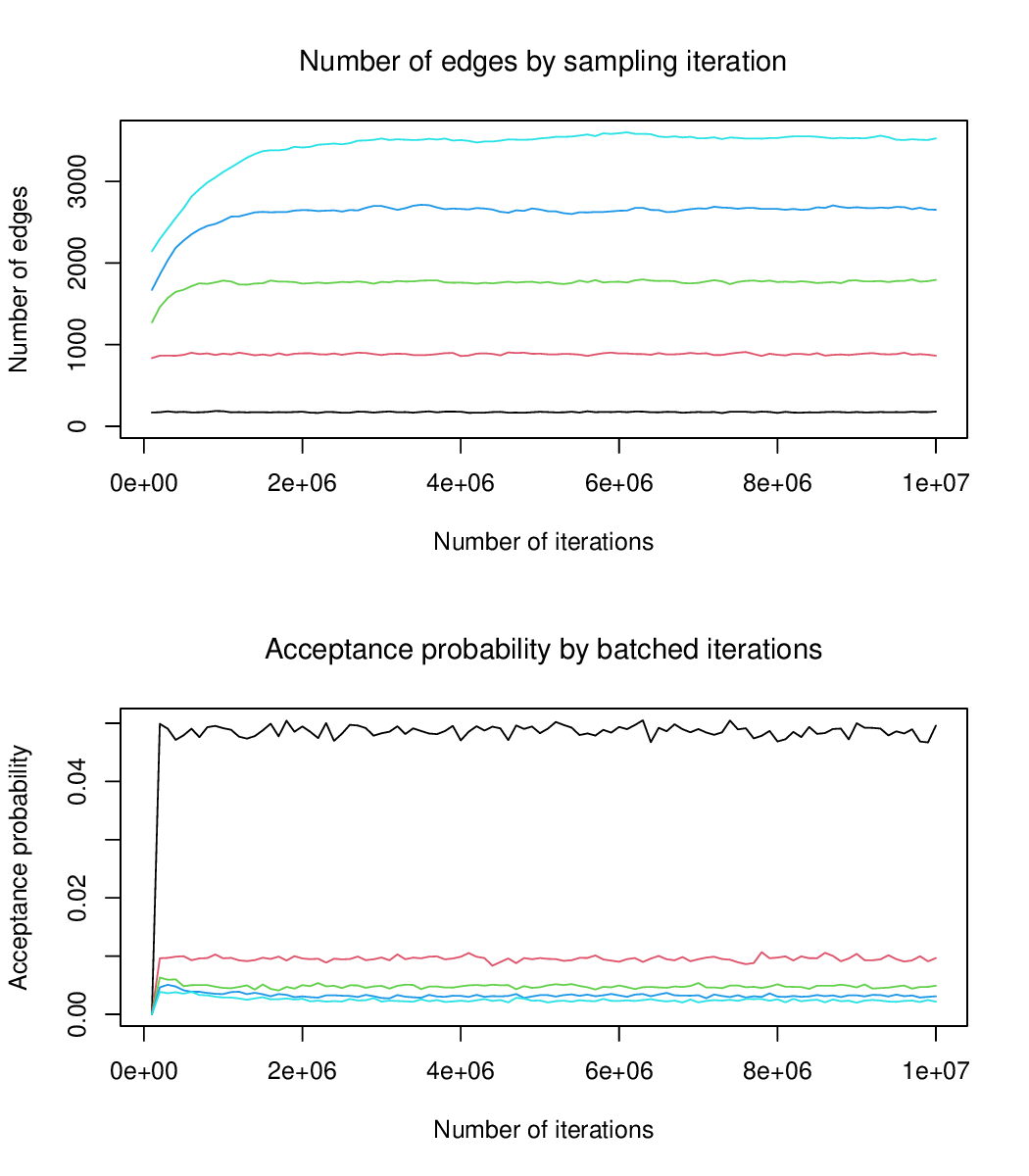} 
\end{knitrout}
\end{figure}

\begin{figure}[htb]
\caption{The 1000000th graph with 100 vertices sampling uniformly over decomposable graphs with maximum clique size 3.
Top left: graph; top right: junction tree; bottom left: Almond tree;
bottom right: Ibarra graph.
\label{egmax3}}
\bigskip
\centerline{
\includegraphics[width=3.0in]{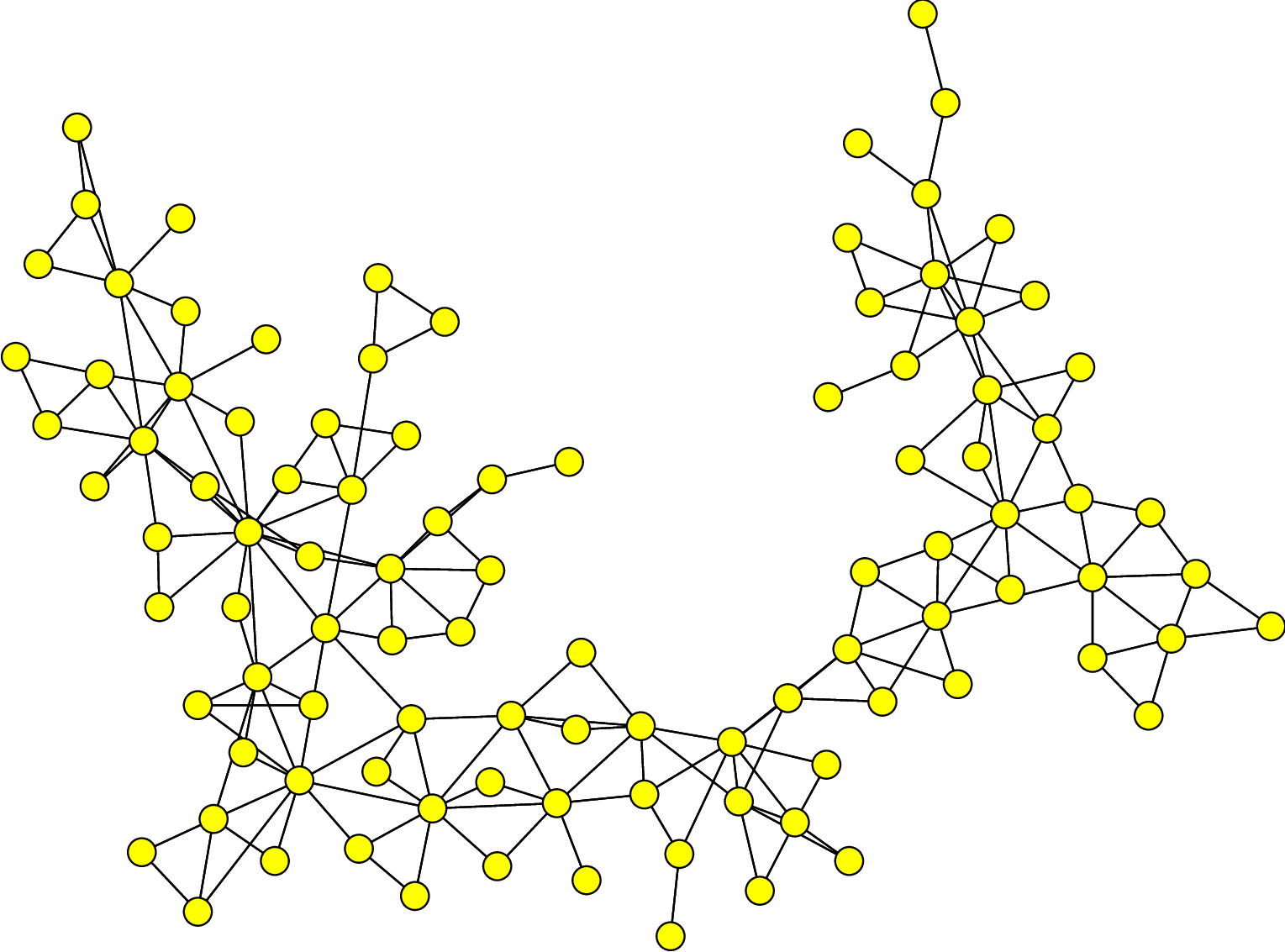}
\includegraphics[width=3.0in]{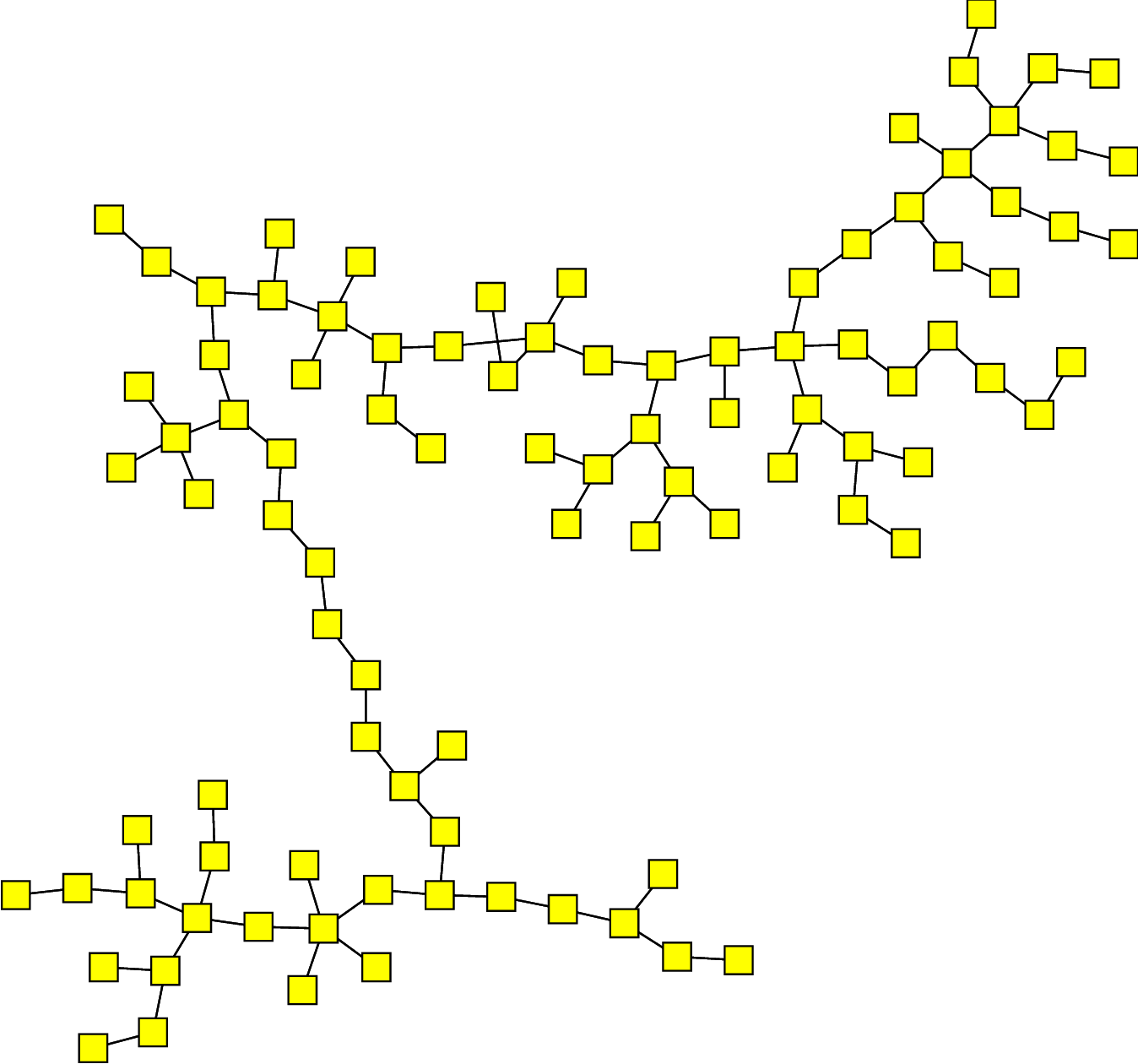}
}
\centerline{
\includegraphics[width=3.0in]{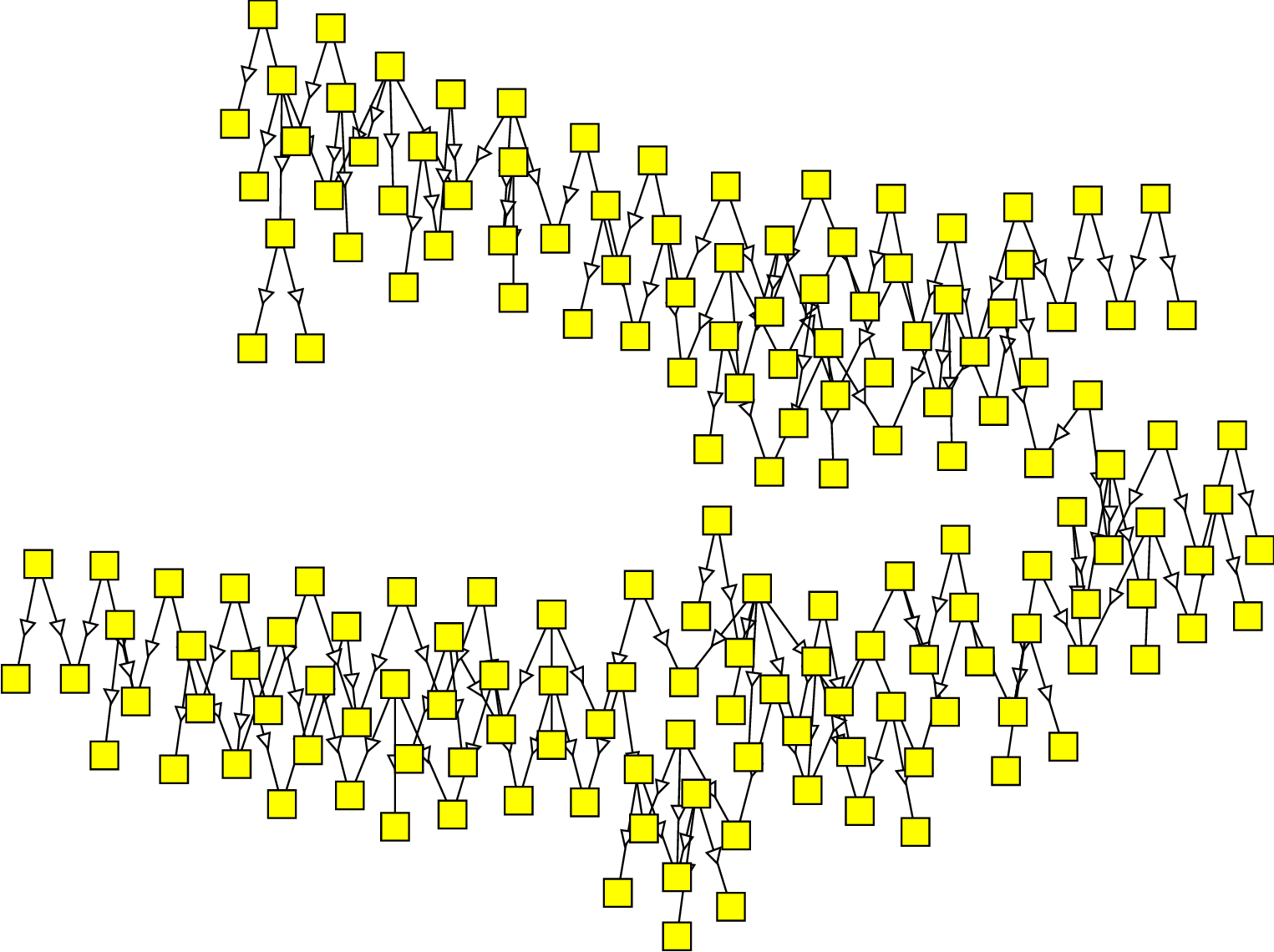}
\includegraphics[width=3.0in]{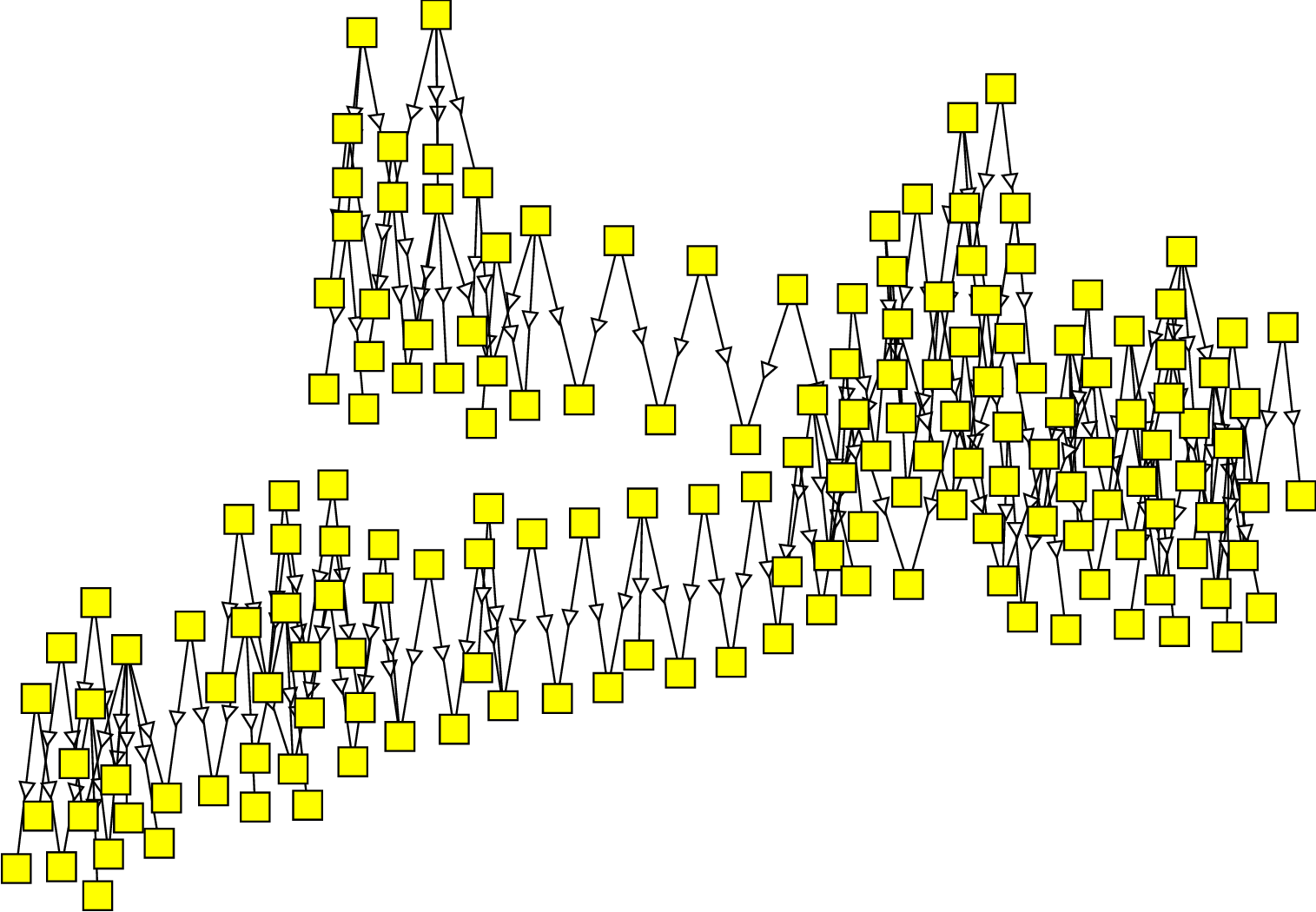}
}
\end{figure}

\begin{figure}[htb]
\caption{
\label{timesalpha1}
Computational times by number of vertices when making
1M and 10M 
samples
over decomposable graphs
with edge penalty 1
using a variety of graph representations.
Black = the graph itself, red = junction tree,
green = Almond tree, blue = Ibarra graph.
}
\bigskip
\begin{knitrout}
\definecolor{shadecolor}{rgb}{0.969, 0.969, 0.969}\color{fgcolor}
\includegraphics[width=\maxwidth]{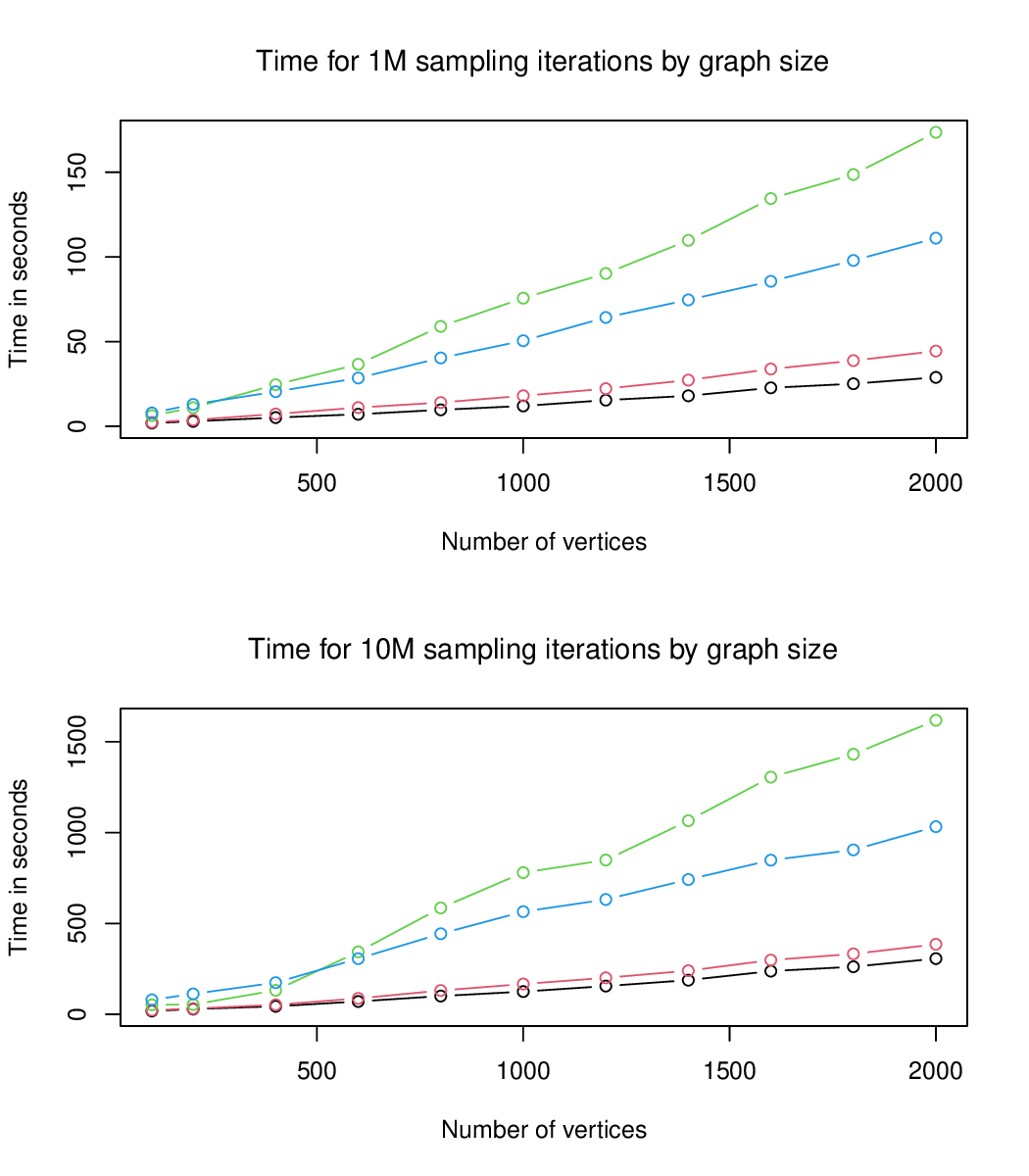} 
\end{knitrout}
\end{figure}

\begin{figure}[htb]
\caption{
\label{edgesa1}
Plots of the number of edges and acceptance probability
by iteration
when making
10M
samples
with edge penalty 1
for a variety of graph sizes.
Black = 100 vertices, red = 500, green = 1000, blue = 1500, cyan = 2000.
}
\bigskip
\begin{knitrout}
\definecolor{shadecolor}{rgb}{0.969, 0.969, 0.969}\color{fgcolor}
\includegraphics[width=\maxwidth]{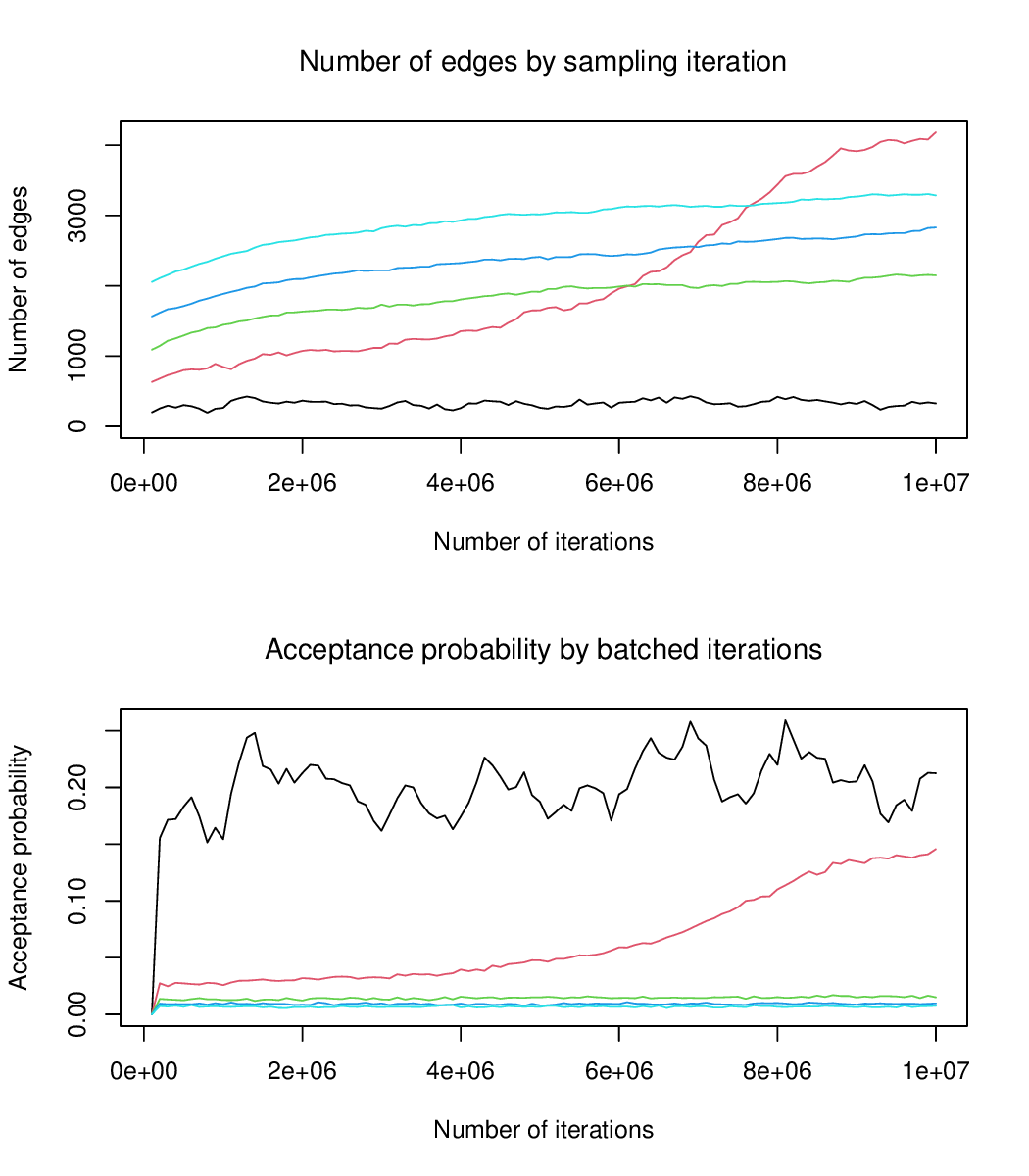} 
\end{knitrout}
\end{figure}

\begin{figure}[htb]
\caption{The 1000000th graph with 100 vertices sampling over decomposable graphs
with edge penalty 1.
Top left: graph; top right: junction tree; bottom left: Almond tree;
bottom right: Ibarra graph.
\label{ega1}}
\bigskip
\centerline{
\includegraphics[width=3.0in]{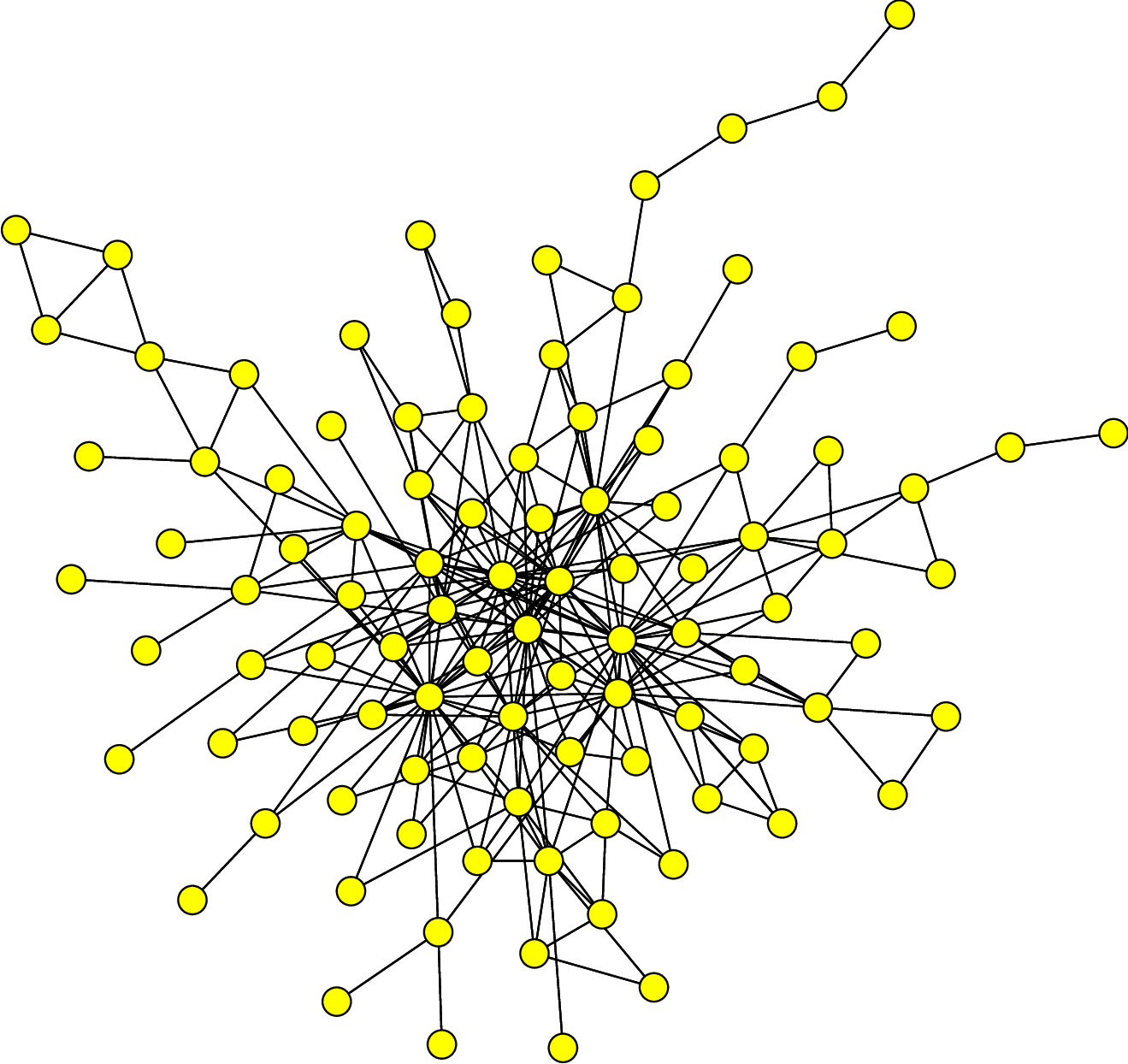}
\includegraphics[width=3.0in]{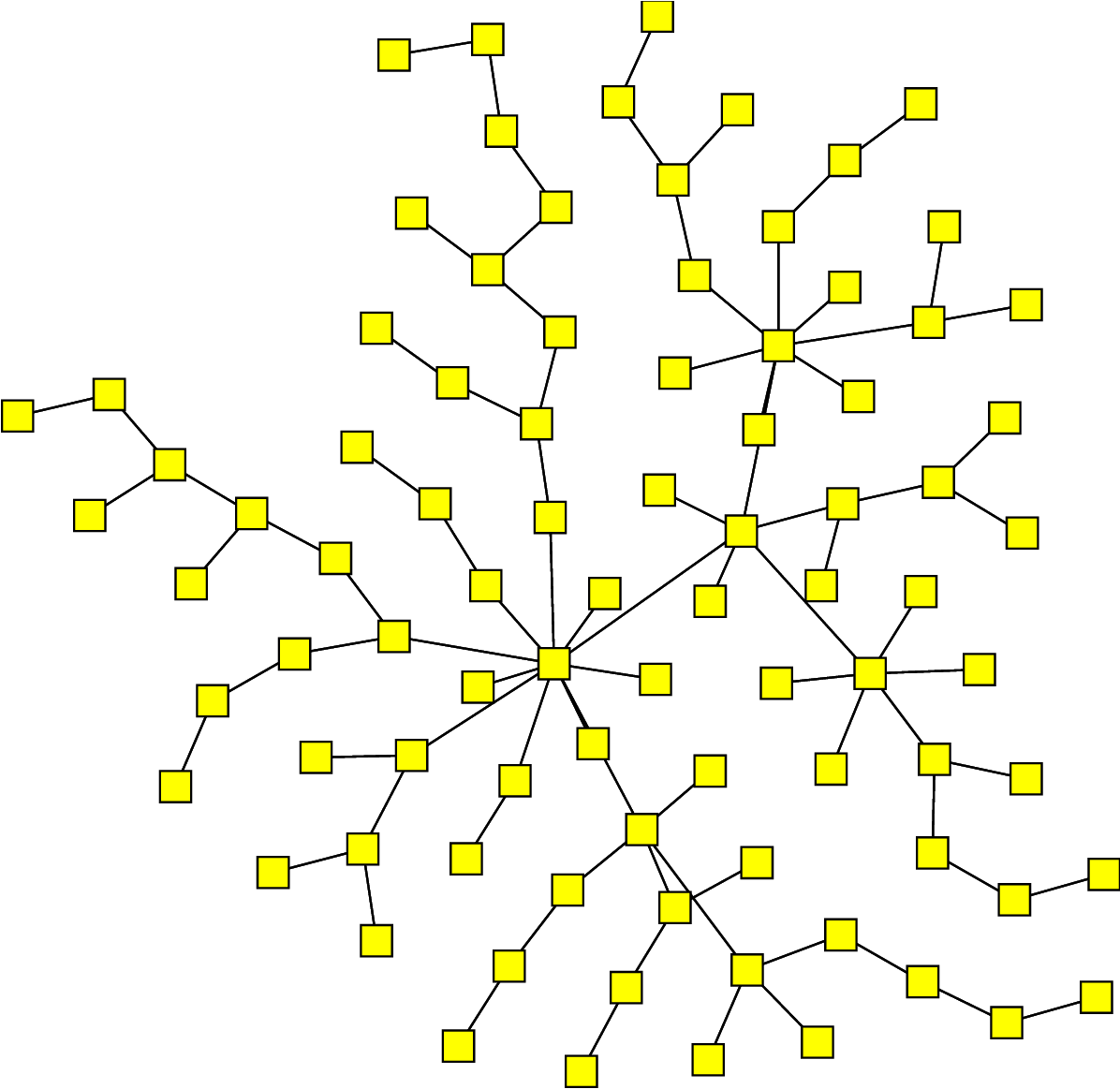}
}
\centerline{
\includegraphics[width=3.0in]{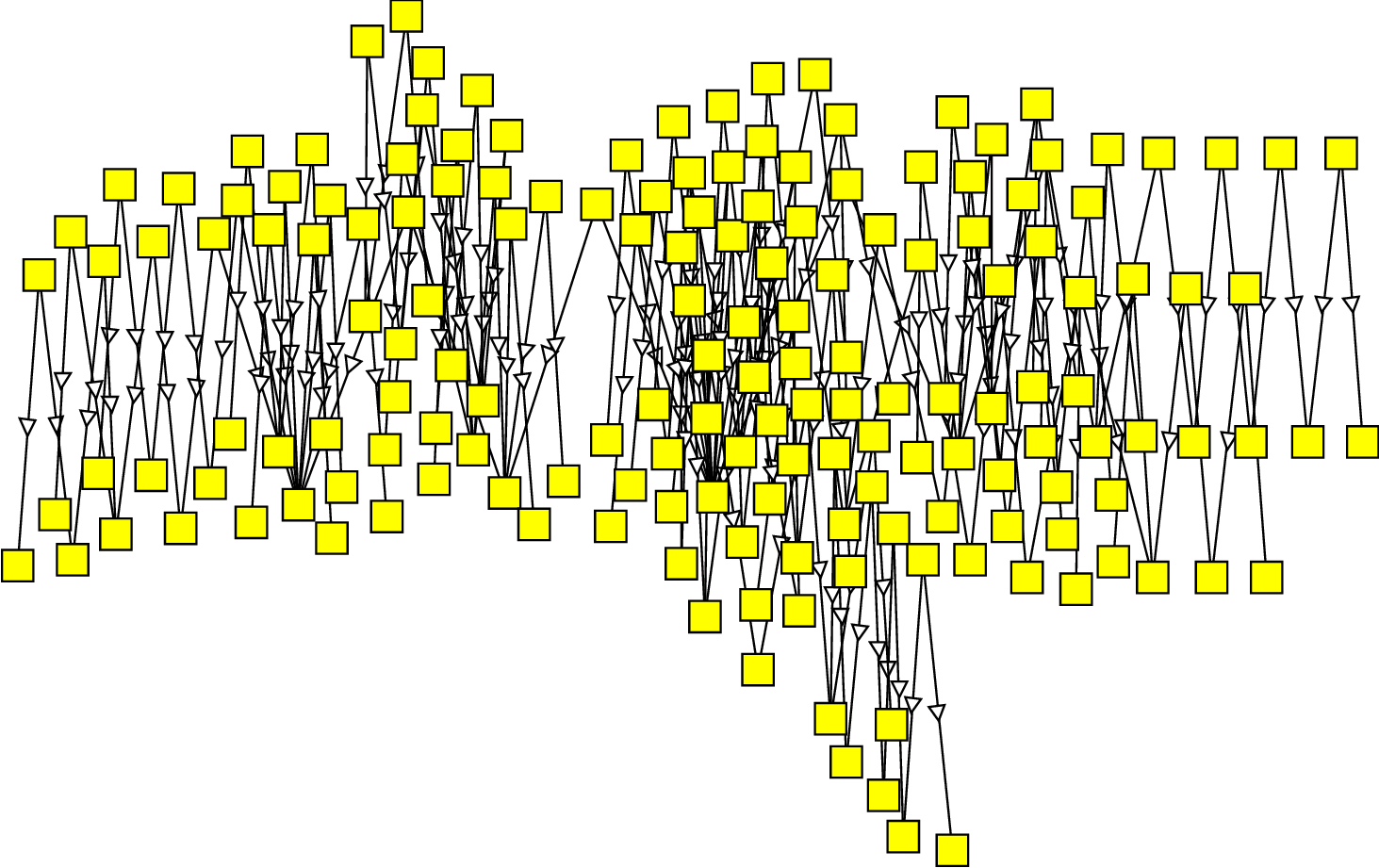}
\includegraphics[width=3.0in]{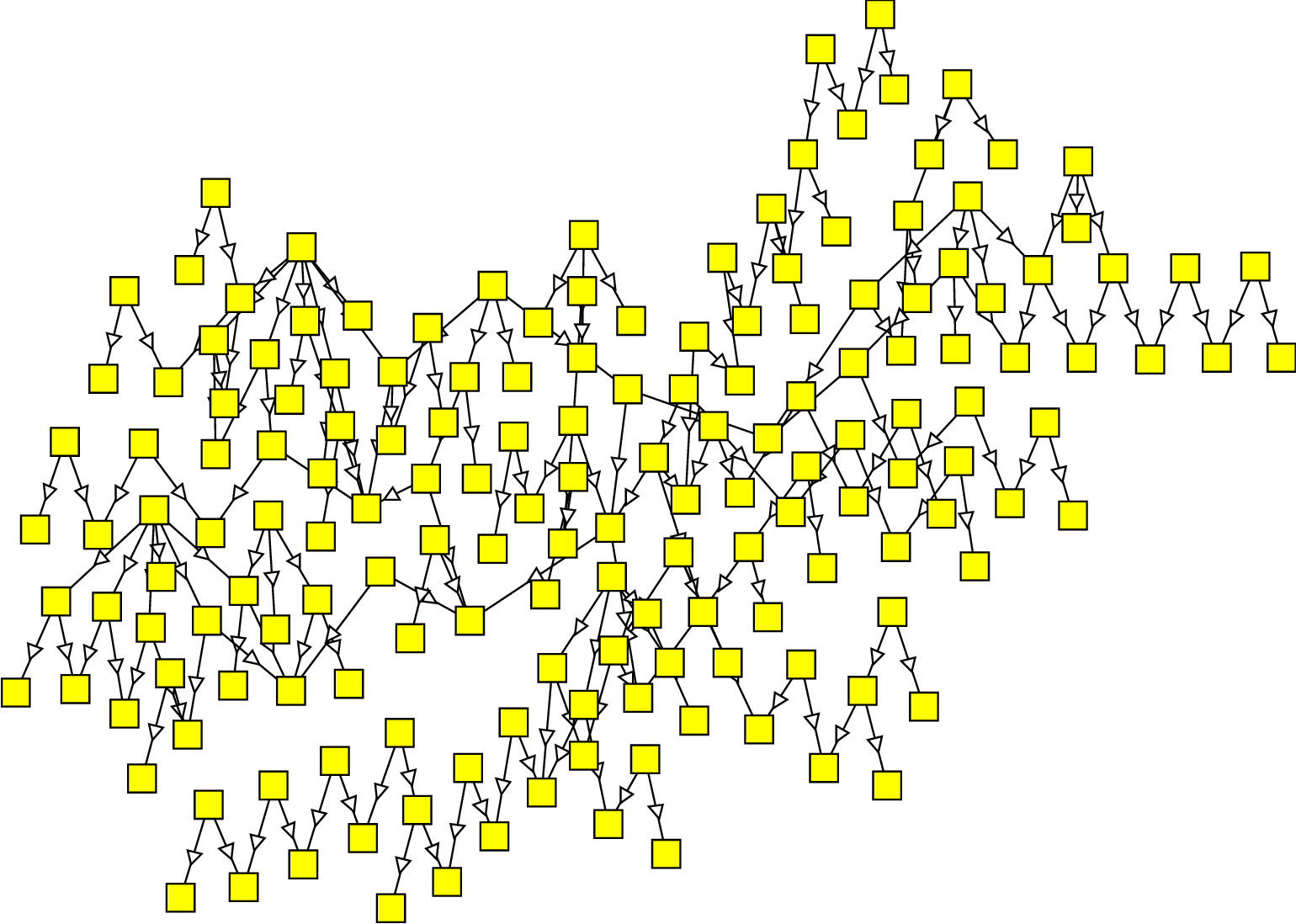}
}
\end{figure}

\begin{figure}[htb]
\caption{
\label{timesalpha2}
Computational times by number of vertices when making
1M and 10M
samples
over decomposable graphs
with edge penalty 2
using a variety of graph representations.
Black = the graph itself, red = junction tree,
green = Almond tree, blue = Ibarra graph.
}
\bigskip
\begin{knitrout}
\definecolor{shadecolor}{rgb}{0.969, 0.969, 0.969}\color{fgcolor}
\includegraphics[width=\maxwidth]{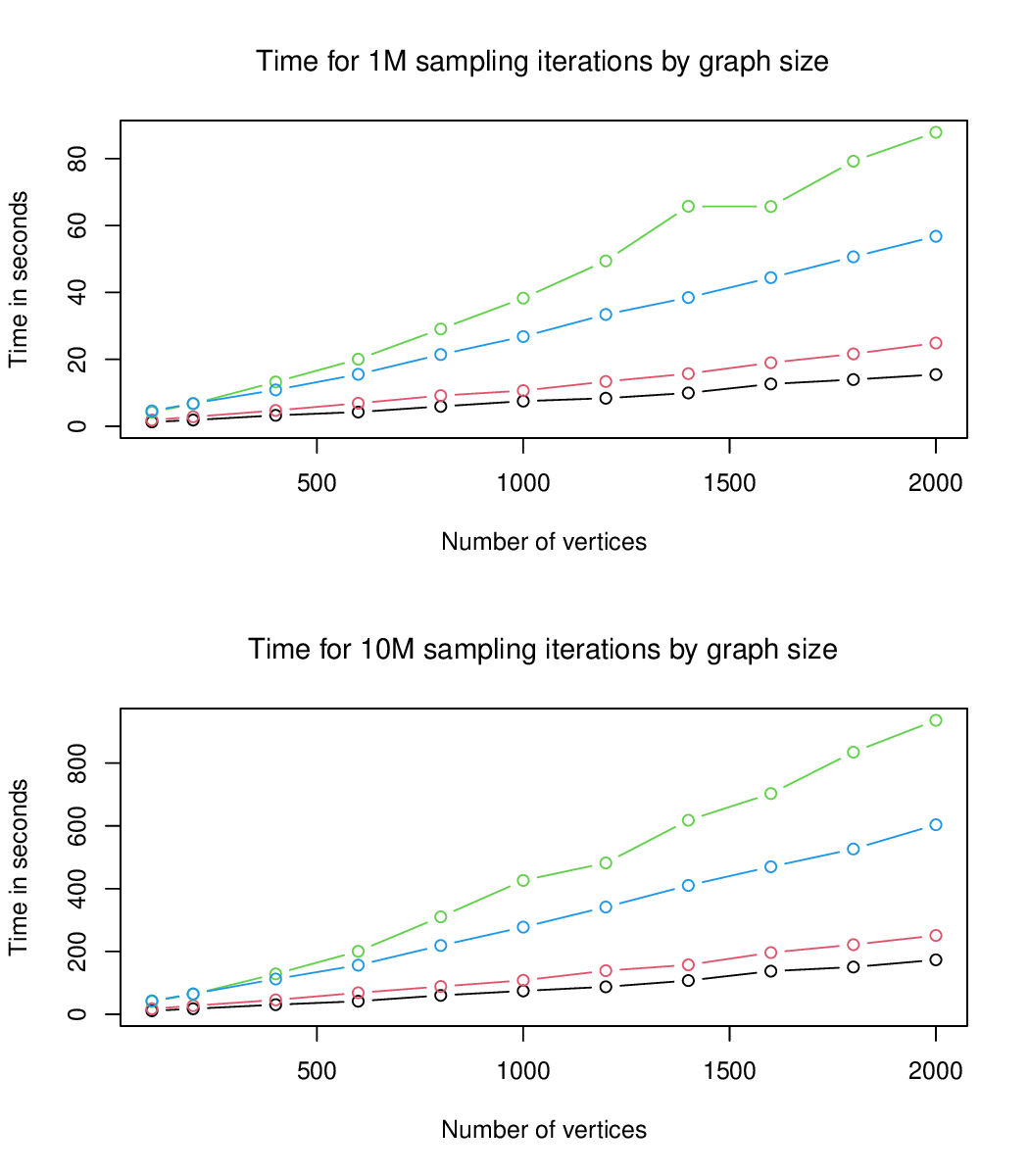} 
\end{knitrout}
\end{figure}

\begin{figure}[htb]
\caption{
\label{edgesa2}
Plots of the number of edges and acceptance probability
by iteration
when making
10M
samples
with edge penalty 1
for a variety of graph sizes.
Black = 100 vertices, red = 500, green = 1000, blue = 1500, cyan = 2000.
}
\bigskip
\begin{knitrout}
\definecolor{shadecolor}{rgb}{0.969, 0.969, 0.969}\color{fgcolor}
\includegraphics[width=\maxwidth]{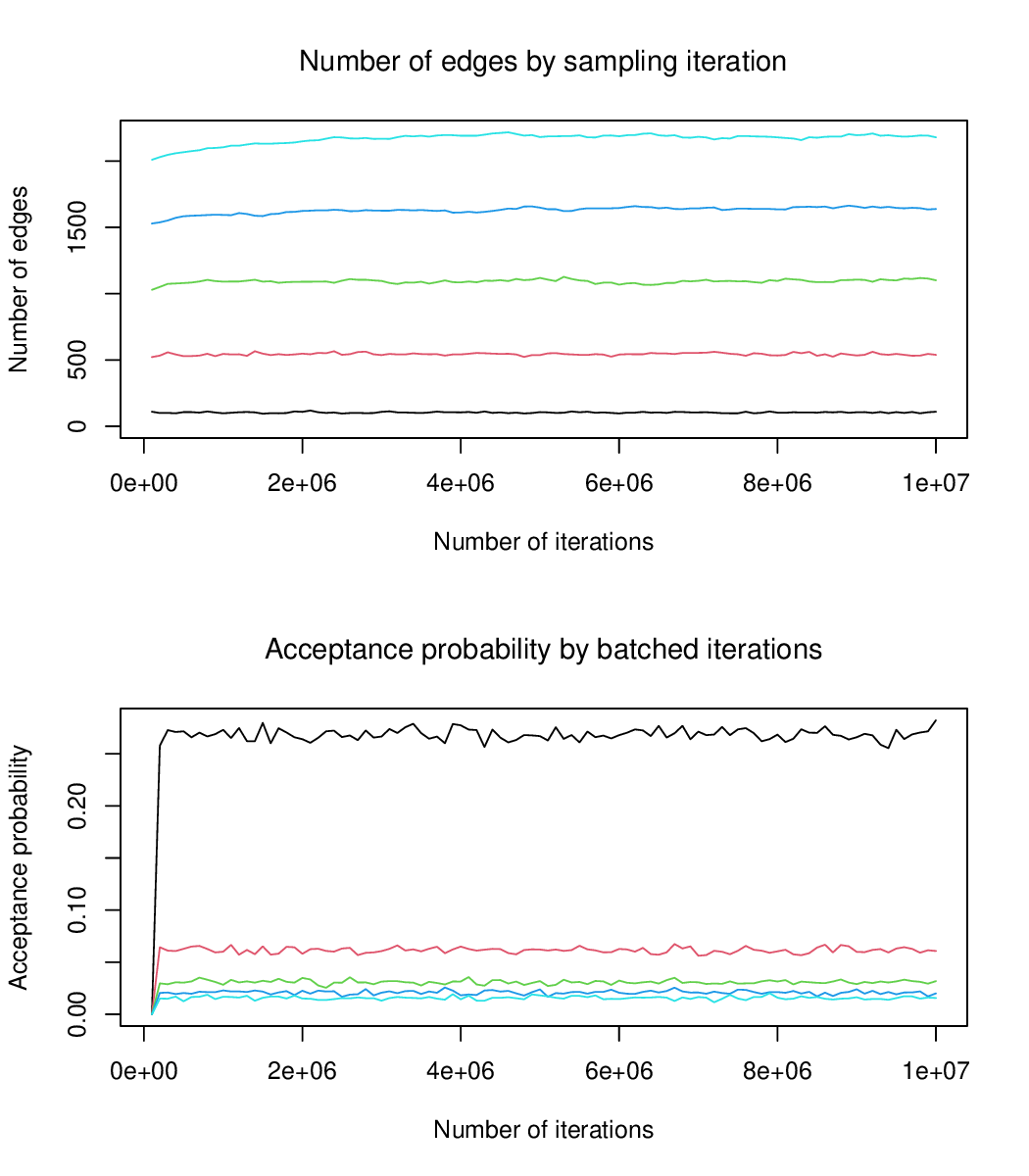} 
\end{knitrout}
\end{figure}

\begin{figure}[htb]
\caption{The 1000000th graph with 100 vertices sampling over decomposable graphs
with edge penalty 2.
Top left: graph; top right: junction tree; bottom left: Almond tree;
bottom right: Ibarra graph.
\label{ega2}}
\bigskip
\centerline{
\includegraphics[width=3.0in]{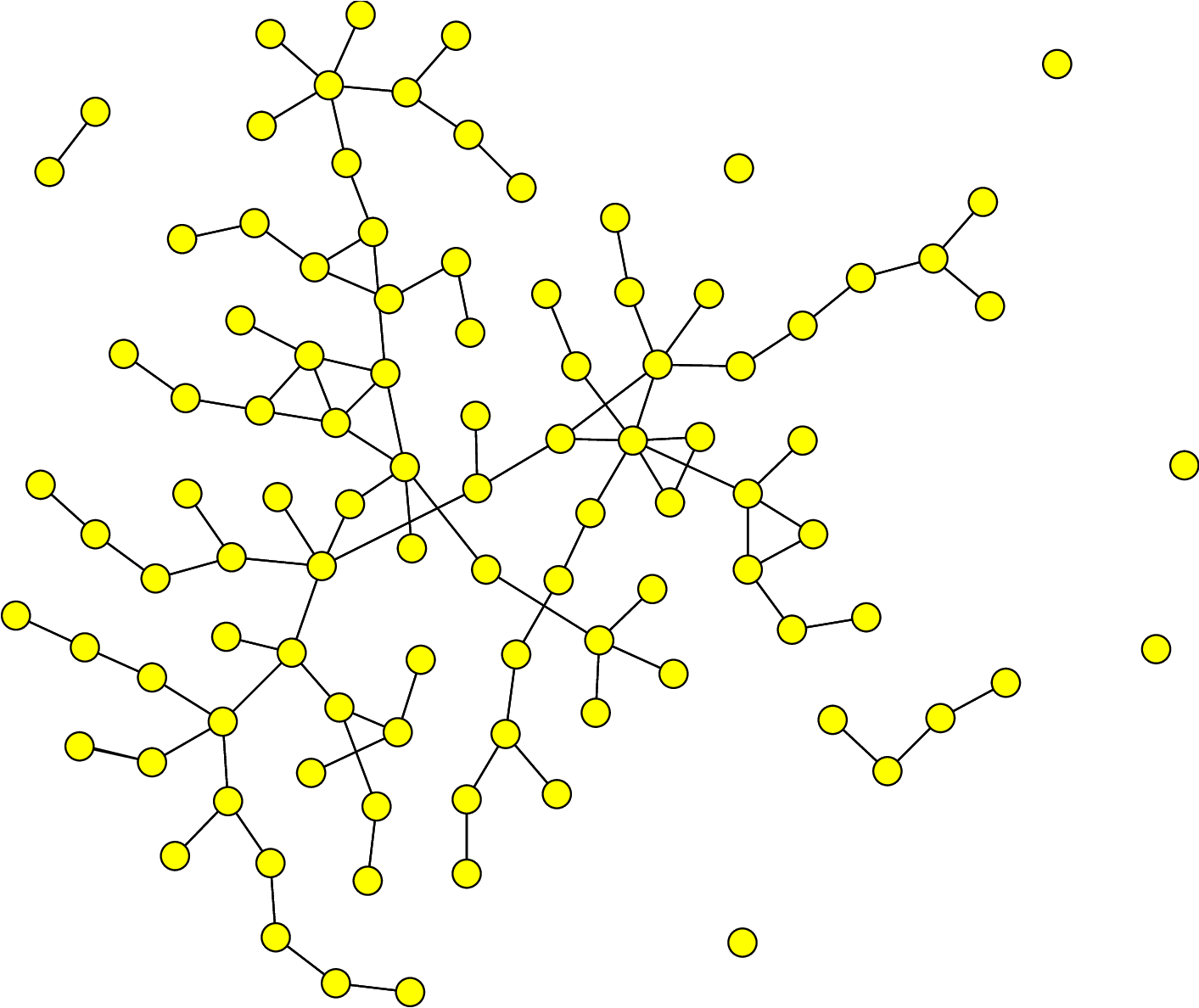}
\includegraphics[width=3.0in]{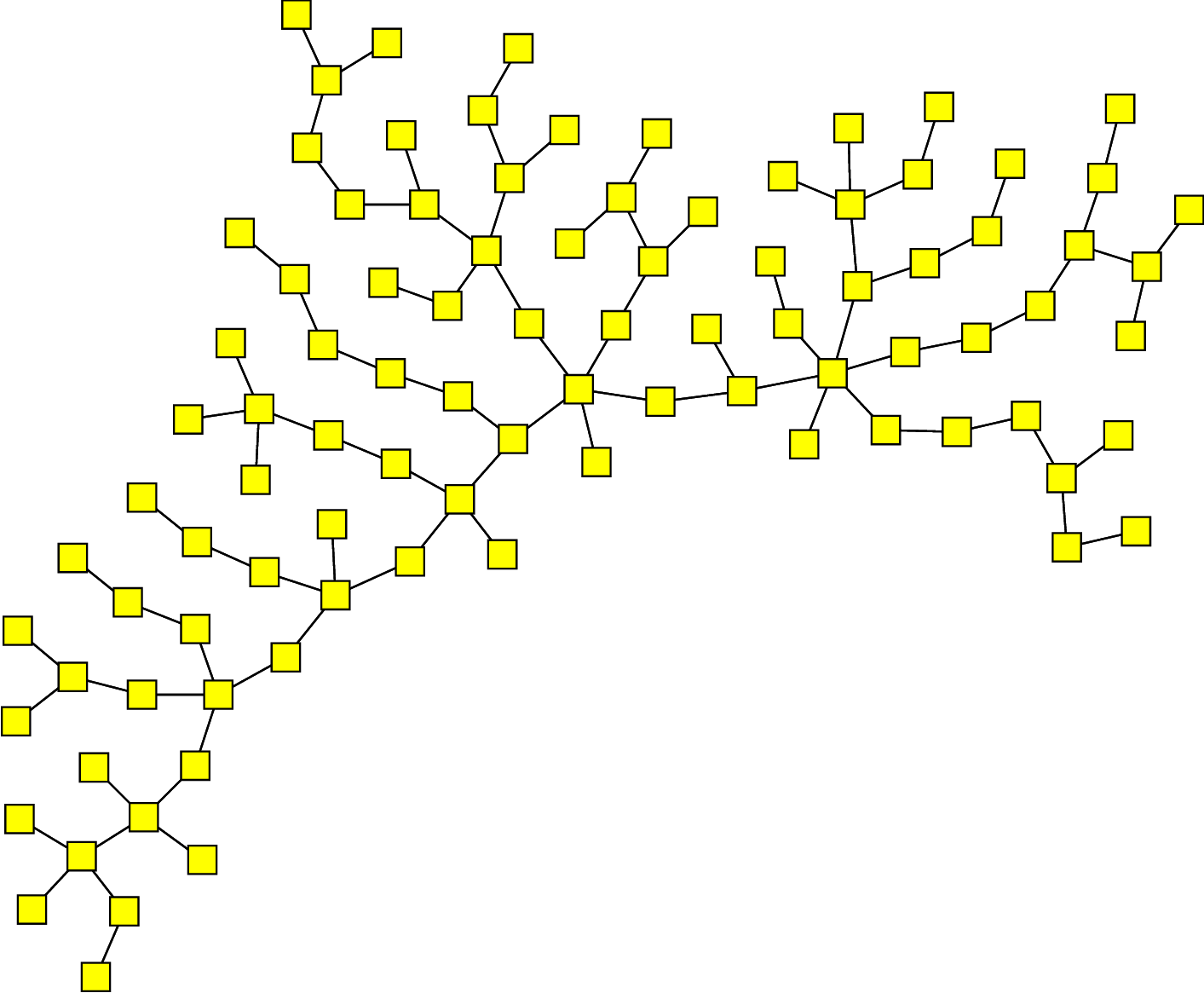}
}
\centerline{
\includegraphics[width=3.0in]{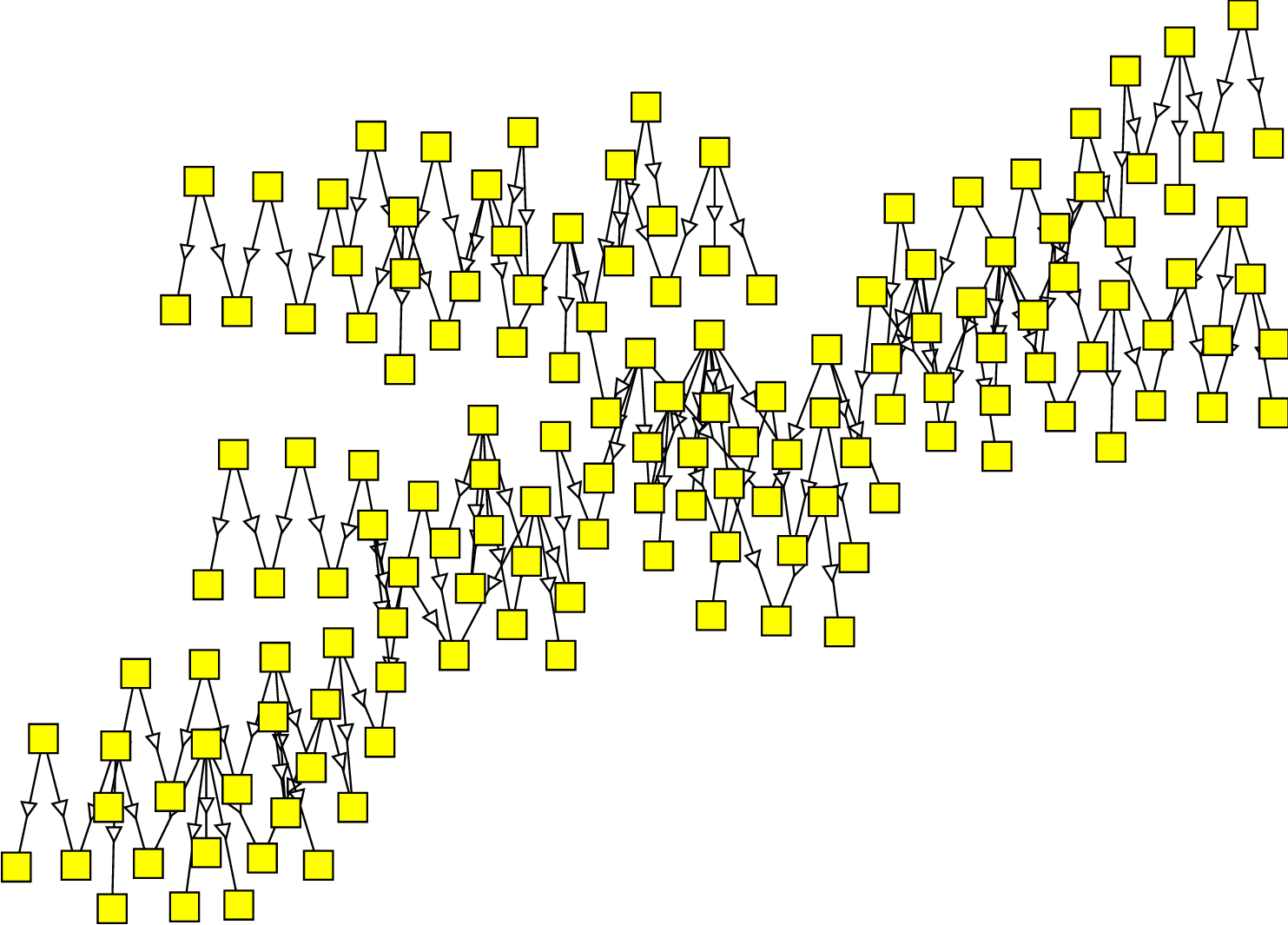}
\includegraphics[width=3.0in]{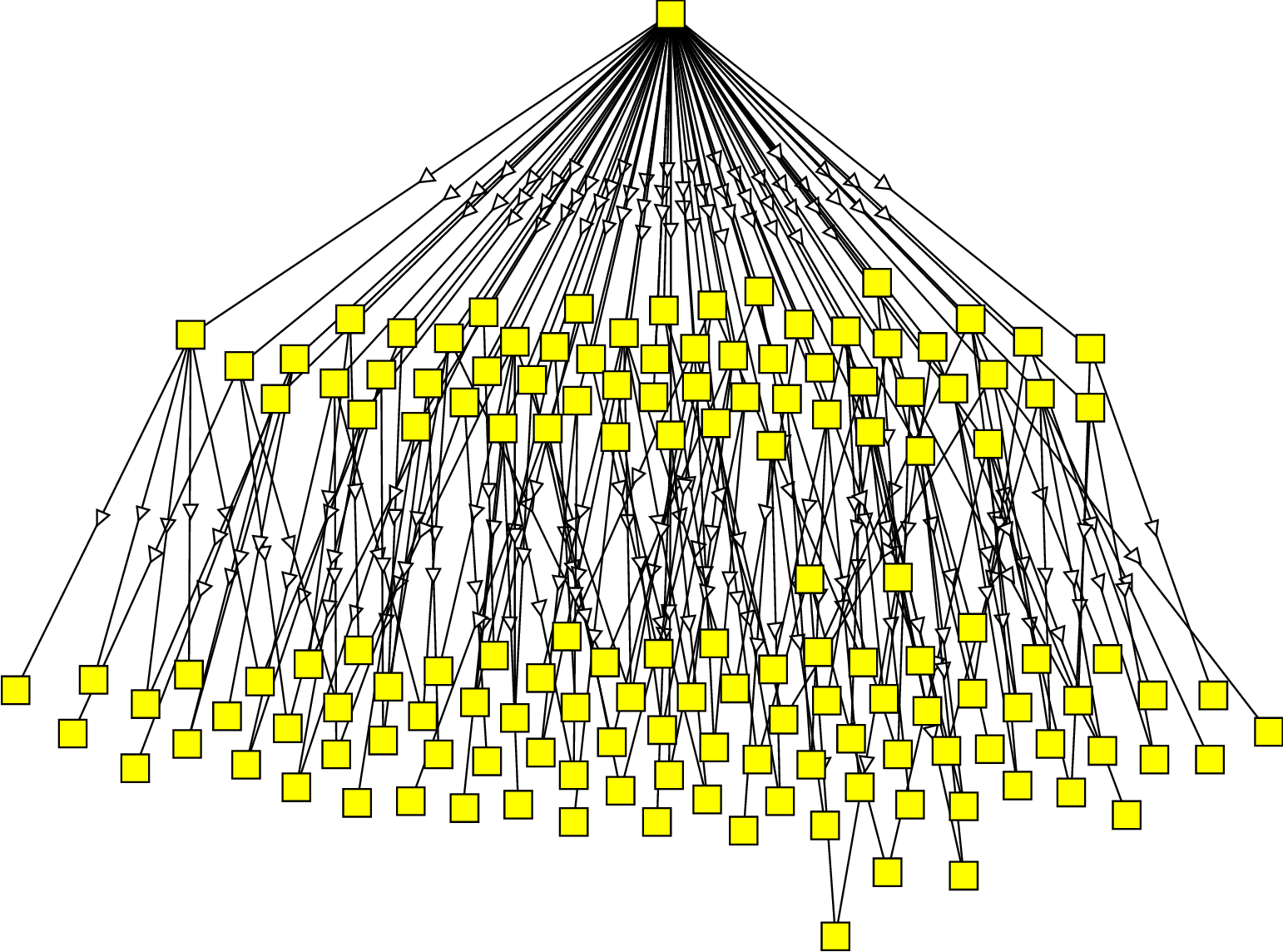}
}
\end{figure}

\clearpage
\section{Discussion}

The first caveat to note is that, obviously, the computational performances 
of the 
different decomposable graph representations depend on the skill, style
and limitations of the programmer and the language and packages used.
Since the Java programs have been made generally available, 
they have been written
more with transparency in mind than for computational optimization, the programs
being close to literal implementations of the algorithms presented here.
Graphs and graph algorithms have been implemented using the same structures
throughout,
essentially collections of adjacency sets represented as hash tables of
hash tables.
The biggest factor to consider here is that the decomposable graph itself 
has simple elements as vertices, whereas the other representations have
vertices which are sets. By default, hash values for sets in Java
are derived from the hash values of their elements, so many queries that
take constant time for the decomposable graph itself may take time linear
in the size of a set for the other representations. 
This can be avoided to a great extent by forcing comparisons of sets
by identity rather than by equality of content, and in other versions
of the code this has been implemented, however, it leads to far less
transparent code that would be more difficult to extend. 

The second note is that the implementations make significant use of
secondary structures or information, and, 
in that sense, are best case scenarios.
In the case of the graph itself, the set of cliques $\C$ and collection
of separators $\S$ are tracked as described in procedure 
\ref{thegraphitself}. Thus, assessing the completeness of $S_{xy}$
is done by querying whether $C_{xy} \in \C$ which is at worst linear 
in $|C_{xy}|$ whereas checking for the presence of each of its 
possible edges in $G$ is quadratic.
Similarly, checking that $S_{xy} \in \S$ can avoid a graph search
to check that it separates
$x$ and $y$ when $S_{xy}$ is not a separator.
Without these savings, the programs take longer, particularly in the
case of sampling uniformly over decomposable graphs which concentrates
on graphs with a few very large cliques.

All the other representations are used in conjunction with the graph itself
and a map from each vertex to a clique that contains it, thus avoiding
list searches to find a containing clique and graph searches to find $S_x$.
The Junction tree also benefits from the pre check that $S_{xy} \in \S$.
Checking that $C_{xy}$ is a clique can be done by querying whether
it is a vertex of any of the set-graphs, and checking that 
$S_{xy}$ is  separator is also just a check that it is a vertex
of either the Almond tree or Ibarra graph.

All representations were run together using the same random proposals
to verify agreement in their assessment of the legality of a
proposal, and by implication, the correctness of the algorithms  given
here.

The primary conclusion from the comparisons
presented above is that there is no overwhelming 
computational reason to use 
any representation other than the graph itself. This approach also has
the attraction of simplicity, given a graph structure that allows
straightforward searching and modification, the algorithm can be 
implemented in very few lines of code. 
However, the Junction tree is in some cases marginally faster
and never much slower and as this representation would benefit far more
from more efficient handling of sets, there is scope for other implementations 
to invert the preference. We should also note that we consider
the graph itself only for model fitting whereas the Junction tree, and
other representations, can be used for message passing forward-backward
algorithms. This strengthens the case for using Junction trees 
if sampling is being used for model averaging any results that require
message passing.

In all cases, for large graphs, Ibarra graphs out perform Almond trees. 
This is likely to be because, in Ibarra graphs, 
several searches, for example finding supersets, require 
following edges in one direction only, whereas more extensive searches
are needed in Almond trees.

A final point to note is that the proportion of proposed perturbations
that yield decomposable graphs decreases rapidly with the size of the
graph dropping to well below 1\% for the largest graphs tried here.
Thus, if any of the representations could be exploited to propose
only legal perturbations, or at least reduce the rejection rate,
this would greatly strengthen the case for their use, and because
of the one-to-one relationship with decomposable graphs, the Ibarra
graph would be prime candidate for investigating this.

\clearpage
\bibliography{collection,alun}
\end{document}